\documentclass[journal]{IEEEtran}
\ifCLASSINFOpdf
\else
\fi

\usepackage{amsthm}

\usepackage{amssymb}
\usepackage{mathtools}
\usepackage{cases}
\usepackage{autobreak}
\usepackage{bm}

\usepackage{booktabs}
\usepackage{multirow}

\usepackage{color}
\usepackage{hyperref}
\usepackage{cite}

\usepackage{amsmath}
\usepackage{extarrows}

\usepackage{comment}

\usepackage{algorithmic}
\usepackage{algorithm}
\usepackage{bbm}

\newcommand{\E}{\mathbb{E}}
\newcommand{\Prob}{\mathbb{P}}

\newcommand{\cov}{\mathrm{Cov}}

\newcommand{\diag}{\mathrm{diag}}

\newcommand{\LR}{\bold{R}^{\frac{1}{2}}}
\newcommand{\RT}{\bold{T}^{\frac{1}{2}}}

\newcommand{\RTh}{\bold{T}^{\frac{H}{2}}}

\newcommand{\FU}{{\bold{U}}}
\newcommand{\MV}{{\bold{V}^{\frac{1}{2}}}}
\newcommand{\FV}{{\bold{V}}}

\newcommand{\MS}{\bold{S}^{\frac{1}{2}}}

\newcommand{\MSl}{\bold{S}^{\frac{+}{2}}}
\newcommand{\MSr}{\bold{S}^{\frac{-}{2}}}
\newcommand{\FSl}{\bold{S}^{+}}
\newcommand{\FSr}{\bold{S}^{-}}

\newcommand{\FR}{\bold{R}}
\newcommand{\FT}{\bold{T}}

\newcommand{\LRk}{{\bold{R}^{\frac{1}{2}}_{k}}}
\newcommand{\RTk}{{\bold{T}^{\frac{1}{2}}_{k}}}

\newcommand{\LRl}{{\bold{R}^{\frac{1}{2}}_{l}}}
\newcommand{\RTl}{{\bold{T}^{\frac{1}{2}}_{l}}}

\newcommand{\FS}{{\bold{S}}}
\newcommand{\BQ}{{\bold{Q}}}
\newcommand{\BH}{{\bold{H}}}
\newcommand{\BI}{{\bold{I}}}
\newcommand{\BX}{{\bold{X}}}

\newcommand{\BZ}{{\bold{Z}}}
\newcommand{\BY}{{\bold{Y}}}

\newcommand{\BG}{{\bold{G}}}

\newcommand{\BA}{{\bold{A}}}
\newcommand{\BB}{{\bold{B}}}
\newcommand{\BC}{{\bold{C}}}
\newcommand{\BP}{{\bold{P}}}
\newcommand{\BW}{{\bold{W}}}
\newcommand{\BM}{{\bold{M}}}

\newcommand{\BF}{{\bold{F}}}

\newcommand{\BO}{{\mathcal{O}}}

\DeclareMathOperator{\Tr}{Tr}

\newcommand{\RNum}[1]{\uppercase\expandafter{\romannumeral #1\relax}}

\newtheorem{remark}{Remark}
\newtheorem{theorem}{Theorem}
\newtheorem{lemma}{Lemma}
\newtheorem{proposition}{Proposition}

\begin{document}
%
\title{Secrecy Analysis for IRS-aided Wiretap MIMO Communications: Fundamental Limits and System Design}
%
%
%

\author{Xin~Zhang,~{\textit{Graduate Student Member,~IEEE}} and 
        Shenghui~Song,~\IEEEmembership{Senior Member,~IEEE}
\thanks{The authors are with the Department of Electronic and Computer Engineering, The Hong Kong University of Science and Technology, Hong Kong
(e-mail: xzhangfe@connect.ust.hk; eeshsong@ust.hk).}}

%
%

\markboth{Journal of \LaTeX\ Class Files,~Vol.~14, No.~8, August~2015}%
{Shell \MakeLowercase{\textit{et al.}}: Bare Demo of IEEEtran.cls for IEEE Journals}
%



\maketitle

\begin{abstract}
In order to meet the demands of future innovative applications, many efforts have been made to exceed the limits predicted by Shannon's Theory. Besides the investigation of beyond-Shannon metrics such as security, latency, and semantics, another direction is to jointly design the transceiver and the environment by utilizing the intelligent reflecting surface (IRS). In this paper, we consider the analysis and design of IRS-aided multiple-input multiple-output (MIMO) secure communications, which has attracted much research attention but still in its infancy. For example, despite their importance, the fundamental limits of IRS-aided wiretap MIMO communications are not yet available in the literature. In this paper, we will investigate these fundamental limits by determining the ergodic secrecy rate (ESR) and secrecy outage probability (SOP). For that purpose, the central limit theorem (CLT) for the joint distributions of the mutual information (MI) statistics over the IRS-aided MIMO secure communication channel is derived by utilizing the random matrix theory (RMT). The derived CLT is then used to obtain the closed form expressions for the ESR and SOP, which are also extended to the scenario with multiple multi-antenna eavesdroppers. Based on the theoretical results, algorithms for maximizing the artificial noise (AN) aided ESR and minimizing the SOP are proposed. Numerical results validate the accuracy of the theoretical results and effectiveness of the proposed optimization algorithms. 
\end{abstract}


%
\IEEEpeerreviewmaketitle

%
%
%
%
\section{introduction}

Shannon's theory built up the foundation of information theory and has been driving the research and development of modern communications systems. However, the traditional Shannon metrics such as throughput and outage probability can no longer meet the requirements of many innovative applications,  e.g., augmented reality/virtual reality (AR/VR) and autonomous driving, which have stringent demands on other metrics including latency and security. Furthermore, semantic communication, which was ignored by Shannon's formulation, has attracted much attention due to the progress of machine learning. Besides the above novel communication metrics, another research direction that is beyond Shannon's theory is to design the channel.  In Shannon's formulation, the channel between the transceiver is determined by the environment and beyond our manipulation. However, the emergence of intelligent reflecting surface (IRS) has totally changed the story where a favorable channel can be constructed for communication, sensing, and energy transfer purposes~\cite{wang2022location,chu2021intelligent}. In this paper, we will investigate the fundamental limits of secure communications over IRS-aided multi-input multi-out (MIMO) channels.     

Wireless communication is vulnerable to eveasdropping due to its broadcast nature. Ensuring the security has been a pivotal issue all through the development of wireless networks. In recent years, besides the upper-layer cryptographic encryption methods, physical layer security (PLS) approaches including cooperative relaying, jamming and artificial noise (AN) have been proposed to guarantee information security~\cite{chen2016survey}. Although these techniques could enhance the secrecy performance, extra power consumption is required for jamming and relaying techniques. Meanwhile, the mere usage of jamming and AN techniques can not always guarantee the secrecy performance in poor propagation environments~\cite{yu2020robust,liu2021secrecy}.

Recently, IRSs have been proposed as a promising solution for future 6G wireless communications, due to its ability to manipulate the direction of the wave by adjusting the reflection phase shifts and offer additional degree of freedom to enhance the connectivity of the network in a low-cost manner. It has been shown in~\cite{guan2020intelligent} that the joint design of the transmit beamforming, the covariance matrix of the AN, and the phase shifts of the IRS outperforms traditional PLS approaches. Motivated by this result, many works have been devoted to the amalgamation of PLS and IRSs~\cite{chu2020secrecy, yu2020robust}. In~\cite{chu2020secrecy}, a joint design of the transmit precoding matrix, covariance matrix of the AN, and IRS phase shifts was proposed to maximize the achievable secrecy rate of the IRS-aided MIMO systems. In~\cite{yu2020robust}, a worst-case robust design of the AN-aided MIMO secure communications with multiple IRSs and multiple Eves was achieved by an alternating optimization (AO) algorithm.

Although many efforts have been dedicated to IRS-aided PLS, there are few works on the characterization of the fundamental performance limits of IRS-aided MIMO secure communications, e.g., the ergodic secrecy rate (ESR) and the secrecy outage probability (SOP), due to the complex structure of the cascaded fading channel, modeled by the product of two random matrices. For example, no result about the SOP analysis of general IRS-aided MIMO systems is available in the literature. In this paper, we will first determine the concerned limits, i.e., ESR and SOP, and then optimize the limits by joint transmitter and IRS design. In the following, we first review the state of arts for the performance characterization and system design of IRS-aided secure communications.  


\subsection{Secrecy Performance Characterization}
In~\cite{yang2020secrecy}, the authors gave an analytical approximation for the SOP of an IRS-aided single-input single-output (SISO) system with one single-antenna eavesdropper by central limit theory (CLT) and showed the positive effect of utilizing IRSs for enhancing the secrecy performance by numerical results. In~\cite{trigui2021secrecy}, the authors obtained the closed-form expressions for the ESR and SOP of IRS-aided SISO systems with discrete phase noise by the Mellin-Barnes integral. The SOP of the SISO system without direct link was investigated in~\cite{do2021secure}. Considering users' locations, the authors of~\cite{zhang2021physical} derived the ESR and SOP of IRS-aided MIMO systems by exploiting stochastic geometry and the  Mellin-Barnes integral, where the spatial correlation at the IRS was not considered. By assuming that the transmit antennas send the same symbol to the user, the authors of~\cite{liu2022minimization} used the Gamma distribution to fit the expressions for the SOP of the IRS-aided wiretap multiple-input multiple-output multiple-antenna-eavesdropper (MIMOME) system. However, the result is not applicable for MIMO systems with general precoding structures.

It is challenging to derive the closed-form performance limits for IRS-aided MIMO secure communications, due to the complex mathematical representation of the secrecy performance. To the best of the authors' knowledge, the ESR with AN and the SOP of general IRS-aided MIMO systems are not yet available in the literature. In this paper, we will give the analytical expressions for the ESR and SOP by large random matrix theory (RMT) for both wiretap and AN-aided MIMO systems by assuming that the number of antennas go to infinity with the same pace. In this asymptotic regime, large RMT has been shown to be powerful for the performance evaluation of MIMO systems~\cite{couillet2011deterministic,hoydis2013massive,hachem2008new,zhang2021bias} as the strikingly simple expression can be obtained and the evaluation is accurate even for small dimensions.

\subsection{System Design for Secure Communications}

Most existing works on secure communication design assumed perfect channel state information (CSI) at the transmitter, However, perfect CSI is extremely difficult to obtain for IRS-aided systems~\cite{kammoun2020asymptotic,zhi2022power}. Meanwhile, the acquisition of CSI for the eavesdroppers at the base station (BS) is also difficult as the eavesdroppers do not interact with the BS frequently~\cite{yu2020robust}. Furthermore, in the scenarios with high mobility users, it is impractical to tune the phase shifts of the IRS to catch up with the changing channels. On the contrary, the statistical CSI varies slowly and can be estimated easily. Furthermore, the design based on statistical CSI does not need to be updated frequently, which not only reduces  the computation workload of the transmitter~\cite{zhao2020intelligent} and the power consumption of the IRS controller, but also releases the overhead of the IRS's control link~\cite{zhi2022power}. In fact, statistical CSI has been widely used in the design of the IRS-aided systems~\cite{zhang2021large,zhang2022outage}. In~\cite{liu2021secrecy}, an AO based algorithm was proposed to maximize the ESR based on statistical CSI.

In this paper, we will first characterize the PLS performance of IRS-aided systems including both wiretap and AN-aided MIMOME systems, based on only statistical CSI. An AO algorithm is then proposed to maximize the AN-aided ESR by jointly optimizing the phase shifts and the transmit covariance matrices of the signal and the AN. To minimize the SOP, we propose a gradient descent algorithm. Finally, numerical results are presented to validate the accuracy of the performance characterization and the effectiveness of the proposed algorithms. The contributions of this paper are summarized as follows.

1) By RMT, we set up the CLTs for the joint distribution of the mutual information (MI) statistics necessary for characterizing the performance of IRS-aided MIMO secure communications. The explicit expression of arbitrary covariances between any MIs are given. The result is more general than the single-variate version given in~\cite{hachem2008new} and~\cite{zhang2022outage} and provides an analytical approach to investigate the joint distribution of MIs with correlated channels.

2) Based on the CLTs, we give closed form expressions for the ESR and SOP of IRS-aided MIMO systems with and without AN, which serve as the fundamental limits of the concerned system. The results are also generalized to the scenario with multiple eavesdroppers. 

3) With only statistical CSI, we propose an algorithm to jointly optimize the transmit covariance matrices (signal and AN) and the phase shifts of the IRS to maximize the ESR. The non-convex problem is handled by successive convex approximation (SCA). Furthermore, we propose an algorithm to minimize the SOP of the IRS-aided system. Simulation results show that the AN-aided approach performs better than that without AN and the proposed algorithms improve the system performance significantly.

\textit{Paper Organizations:} We organize the rest of the paper as follows. Section~\ref{sys_mol} introduces the system model and formulates the problem. Section~\ref{sec_clts} presents the theoretical results----CLT for the joint distribution of the MIs. Section~\ref{sec_esr} and Section~\ref{sec_sop} give the analytical expressions of the ESR and SOP for the wiretap and AN-aided systems, respectively. Section~\ref{sec_alg} provides the AO algorithm to maximize the ESR by jointly designing the phase shifts and the covariance matrices of the signal and AN. A gradient based algorithm is also proposed to minimize the SOP. Section~\ref{sec_simu} presents numerical results to illustrate the accuracy of the analytical results and the performance of the proposed algorithms. Section~\ref{sec_con} concludes this paper.

\textit{Notations:} We use the bold, upper case letters and bold, lower case letters to denote matrices and vectors, respectively. $\mathbb{P}(\cdot)$ represents the probability operator and $\E x$ denotes the expectation of $x$. $\mathbb{C}^{N}$ and $\mathbb{C}^{M\times N}$ represent the $N$-dimensional vector space and the $M$-by-$N$ matrix space, respectively. $(\cdot)^{*}$ represents the conjugate of a complex number. $\bold{A}^{H}$ represents the conjugate transpose of $\bold{A}$, and the $(i,j)$-th entry of $\bold{A}$ is denoted by or $[\BA]_{i,j}$ or $A_{ij}$. $\|\BA \|$ represents the spectral norm of $\bold{A}$. $\Tr\BA$ refers to the trace of $\BA$ if it is square and $\otimes$ denotes the element-wise product of matrices. $\bold{I}_{N}$ denotes the $N$-dimension identity matrix. The cumulative distribution function (CDF) of the standard normal distribution is denoted by $\Phi(x)$. $\mathbbm{1}_{e}$ denotes the indicator function, i.e., $\lceil \cdot \rceil^{+}$ represents the ceiling function. $\mathbbm{1}_{e}=1$ if $e$ holds true and else $\mathbbm{1}_{e}=0$. $\underline{x}=x-\E x$ represents the centered random variable and $\cov(x,y)=\E\underline{x}\underline{y} $ denotes the covariance of random variables $x$ between $y$. $\xrightarrow[N \rightarrow \infty]{\mathcal{D}}$ denotes the convergence in distribution.  


\section{System Model}
\label{sys_mol}
Consider an IRS-aided downlink MIMO system, consisting of a base station (Alice or BS), a user (Bob), and an eavesdropper (Eve), which are equipped with $M$, $N_{B}$, and $N_{E}$ antennas, respectively. The IRS has $L$ elements. Due the to blockage, there is no direct link between the BS and Bob, and the same happens to Eve. The BS-IRS channel is denoted by $\BH_{T}\in \mathbb{C}^{L \times M}$. The IRS-Bob and IRS-Eve channels are represented by $\BH_{B,I}\in  \mathbb{C}^{N_B \times L}$ and $\BH_{E,I}\in  \mathbb{C}^{N_E \times L}$, respectively. In the following, we will introduce two secure communication schemes, i.e., wiretap systems and AN-aided systems. 

\subsection{IRS-aided Wiretap Systems}
In the wiretap system (AN is not considered here), the received signal $\bold{y}_{B}\in \mathbb{C}^{N_B}$ of Bob is given by 
\begin{equation}
\bold{y}_{B}=\BH_{B,I}\bold{\Theta}\BH_{T}\BW\bold{x}+\bold{n}_{B},
\end{equation}
where $\bold{x}\in \mathbb{C}^{M} $, following $\mathcal{CN}(\bold{0} ,\bold{I}_{M})$, represents the transmitted signal. $\BW$ represents the precoding matrix at the BS. $\bold{n}_{B}\in \mathbb{C}^{N_{B}}$ is the additive white Gaussian noise (AWGN), with variance $\sigma^2_{B}$ and $\bold{\Theta}=\diag(e^{\jmath \theta_{1}},e^{\jmath \theta_{1}},...,e^{\jmath \theta_{L}})$ denotes the phase shifts introduced by the IRS. Similarly, the received signal of Eve $\bold{y}_{E}\in \mathbb{C}^{N_{E}}$ is represented by
\begin{equation}
\bold{y}_{E}=\BH_{E,I}\bold{\Theta}\BH_{T}\BW\bold{x}+\bold{n}_{E}.
\end{equation}
The cascaded channel of Bob and Eve are denoted by $\BH_{B}=\BH_{B,I}\bold{\Theta}\BH_{T}$ and $\BH_{E}=\BH_{E,I}\bold{\Theta}\BH_{T}$, respectively.

For ease of illustration, we introduce the following MI related notation
\begin{equation}
\begin{aligned}
M(z_k,\BH_k,\BP_k)=\log\det(z_k\bold{I}_{N_{k}}+\BH_{k}\BP_k\BH^{H}_{k}).
\end{aligned}
\end{equation}
Assuming Wyner code is utilized for encoding the confidential message over an ergodic fading channel, the achievable ESR (bit/s/Hz) of the wiretap system is given by~\cite{gopala2008secrecy}
\begin{equation}
\label{esr_def}
\begin{aligned}
&M_{S}(\BP_{W})
=\lceil\E F_{B,E}(\BP_W)
-N_B\log(\sigma^2_B)
\\
&
+N_E\log(\sigma^2_E) \rceil^{+}
,
\end{aligned}
\end{equation}
where $\BP_W=\BW\BW^{H}$ and $F_{B,E}(\BP_W)=M(\sigma^2_B,\BH_B,\BP_W)-M(\sigma^2_E,\BH_E,\BP_W).
$
Given a rate threshold $R$, the SOP of the wiretap system is given by
\begin{equation}
\begin{aligned}
&P_{S,out}(R)=\Prob(F_{B,E}(\BP_W)
\\
&
-N_B\log(\sigma^2_B)+N_E\log(\sigma^2_E) <R).
\end{aligned}
\end{equation} 

\subsection{Artificial Noise-aided Systems}
With AN, the transmitted
signal can be modeled as
\begin{equation}
\bold{s}=\BW\bold{x}+\bold{n}_{A},
\end{equation}
where $\BW\in \mathbb{C}^{M\times M}$ represents the precoding matrix at the BS and $\bold{n}_{A}\in \mathbb{C}^{M}$ denotes the AN, which follows the Gaussian distribution $\mathcal{N}(\bold{0}_{M},\BP_V)$. In this case, the received signals of Bob and Eve are given by
\begin{equation}
\begin{aligned}
\bold{y}_{B}=\BH_{B}(\BW\bold{x}+\bold{n}_{A})+\bold{n}_{B},
\\
\bold{y}_{E}=\BH_{E}(\BW\bold{x}+\bold{n}_{A})+\bold{n}_{E}.
\end{aligned}
\end{equation}
Define
\begin{equation}
F_{I,B,E}(\BP_W,\BP_V)=M_{I,B}(\BP_W,\BP_V)-M_{I,E}(\BP_W,\BP_V),
\end{equation} 
where
\begin{equation}
\label{C_inf}
M_{I,k}(\BP_1,\BP_2)=\log\det(\bold{I}_{N_{k}}+\BH_{k}\BP_1 \BH^{H}_{k}\FU_{k}(\BP_2)),
\end{equation} 
and $\FU_{k}(\BP_2)=\left(\sigma_{k}^{2}\bold{I}_{N_{k}}+\BH_{k}\BP_2 \BH_{k} \right)^{-1}$ is the information-plus-noise covariance matrix. The ESR of the AN-aided system can be expressed as~\cite{gopala2008secrecy,wang2016artificial}
\begin{equation}
\label{ESR_AN}
\begin{aligned}
 M_{AN}(\BP_W,\BP_V)
=
\lceil \E [F_{I,B,E}(\BP_W,\BP_V)] \rceil^{+}.
\end{aligned}
\end{equation}
Given a rate threshold $R$, the SOP of the AN-aided system is given by
\begin{equation}
P_{AN,out}(R)=\Prob({F}_{I,B,E}(\BP_{W},\BP_V) <R).
\end{equation} 

\subsection{Channel Model}
In this paper, we consider two types of channel and use the Kronecker model to characterize the spatial correlation for each MIMO link.
\subsubsection{LoS BS-IRS (LBI) case}
In this case, the BS-IRS link is dominated by the line-of-sight (LoS) component and the IRS-Bob and IRS-Eve links are modeled as a correlated Rayleigh channel. Under such circumstance, the channel matrix can be given by
\begin{equation}
\label{singl_ch}
\BH_{k}=\LR_{k}\BX_{k}\RT_{S,k}\bold{\Theta}\BH_{T,0},~~k=B,E.
\end{equation} 
where $\BH_{T,0}\in\mathbb{C}^{M \times L}$ represents the LoS channel from the BS to the IRS. $\FT_{S,k}$ and $\FR_{k}$ denote the spatial correlation matrices at the IRS and the receiver of user $k$. $\bold{X}_{k}\in\mathbb{C}^{N _k\times L}$ denotes an independent and identically distributed (i.i.d.) Gaussian random matrix, whose entries follow $\mathcal{CN}(0, \frac{1}{L})$. The BS-IRS channel is regarded as a deterministic channel such that $\BH_{k}$ can be treated as a single Rayleigh channel with a special transmit correlation matrix, i.e., $\BH_k=\LR_k\BX_k\FT^{\frac{+}{2}}_{k}$, where $\FT^{\frac{+}{2}}_{k}=\RT_{S,k}\bold{\Theta}\BH_{T,0}$ and $\FT^{\frac{-}{2}}_{k}=(\FT^{\frac{+}{2}}_{k})^{H}$ for $k=B,E$. As a result, we will also refer to this case as the single hop case~\cite{kammoun2020asymptotic}.
\subsubsection{Double-scattering case}
For the double scattering case, the equivalent channel between the BS and user $k$ can be given by 
\begin{equation}
\BH_{k}=\BH_{k,1}\bold{\Theta}\BH_{k,2},
\end{equation}
where
\begin{equation}
\label{cha_mod}
\bold{H}_{k,1}=\bold{R}_{k}^{\frac{1}{2}}\bold{X}_{k}\bold{T}^{\frac{1}{2}}_{S,k},~\bold{H}_{k,2}=\bold{R}_{S}^{\frac{1}{2}}\bold{Y}\bold{T}^{\frac{1}{2}}
\end{equation}
represent the channel from the IRS to user $k$ and the channel from the BS to the IRS, respectively. $\bold{R}_{k}\in \mathbb{C}^{N_{k}\times N_k}$, $\bold{T}_{k,S}\in \mathbb{C}^{L\times L}$, $\FR_{S}\in  \mathbb{C}^{L \times L}$ and $\FT\in \mathbb{C}^{M \times M}$ are positive semi-definite matrices. In particular, $\bold{T}_{S,k}$ and $\FR_{S}$ represent the transmit and receive correlation matrix of the IRS. $\bold{R}_{k}$ denotes the spatial correlation matrix at user $k$ ($k=B,E$ for Bob or Eve) and $\FT$ denotes the correlation at the BS. $\bold{X}_k \in\mathbb{C}^{N_k \times L}$ and $\bold{Y}\in\mathbb{C}^{L \times M}$ are independent and identically distributed (i.i.d.) Gaussian random matrices, whose entries follow $\mathcal{CN}(0, \frac{1}{L})$ and $\mathcal{CN}(0, \frac{1}{M})$, respectively. We assume that statistical CSI, i.e., correlation matrices of the channel, is available. 
Denote $\FS^{\frac{+}{2}}_k=\bold{T}^{\frac{1}{2}}_{S,k}\bold{\Theta}\LR_{S}$, $\FS^{\frac{-}{2}}_k=(\FS^{\frac{+}{2}})^{H}$ and $\FS_k=\FS_k^{\frac{-}{2}}\FS^{\frac{+}{2}}$. The equivalent channel can also be represented by
\begin{equation}
\label{doub_ch}
\BH_{k}=\LR_k \BX_k \FS^{\frac{+}{2}}_k\BY\RT.
\end{equation}

\subsection{Problem Formulation}
In this paper, we will investigate the asymptotic characterization of the ESR and SOP of IRS-aided MIMOME systems. The secrecy rate of the wiretap system and AN-aided system in~(\ref{esr_def}) and (\ref{ESR_AN}) can be rewritten as
\begin{subequations}
\label{c_inf_ab}
\begin{align}
\label{C_inf1}
&F_{B,E}(\BP_W)=M(\sigma^2_B,\BH_B,\BP_W)-M(\sigma^2_E,\BH_E,\BP_W),
\\
\begin{split}
\label{AN_MIS}
&F_{I,B,E}(\BP_W,\BP_V)=M(\sigma^{2}_B,\BH_B,\BP_{U})
\\
&
-M(\sigma^{2}_B,\BH_B,\BP_{V})
-M(\sigma^{2}_E,\BH_E,\BP_{U})
\\
&+M(\sigma^{2}_E,\BH_E,\BP_{V}),
\end{split}
\end{align}
\end{subequations}
respectively, where $\BP_{U}=\BP_{W}+\BP_V$. It can be observed that the secrecy rates of both wiretap and AN-aided systems can be represented by linear combinations of the MI statistics $M(z_k,\BH_k,\BP)$. Therefore, the characterization of the ESR and SOP can be resolved if we can obtain the joint distribution of MIs in~(\ref{C_inf1}) and~(\ref{AN_MIS}). 

The two MIs in~(\ref{C_inf1}) can be used to quantify the information received by Bob and Eve with the same transmit covariance $\BP_W$. The first MI in~(\ref{AN_MIS}) can be regarded as the MI between the BS and Bob with the transmit matrix $\BP_U$ and the second MI represents the loss induced by the interference of the AN with covariance matrix $\BP_V$. The last two MIs in~(\ref{AN_MIS}) can be understood similarly. The MIs involved in~(\ref{AN_MIS}) have two types of relations. On one hand, there are MIs that have independent $\BX$, e.g., $M(\sigma^{2}_B,\BH_B,\BP_{U})$ and $M(\sigma^{2}_E,\BH_E,\BP_{U})$. On the other hand, there are MIs that share the same $\BX$. For example, $M(\sigma^{2}_B,\BH_B,\BP_{U})$ and $M(\sigma^{2}_B,\BH_B,\BP_{V})$ share the same $\BX_{B}$ and the only difference comes from the transmit covariance matrix. As a result, to investigate the joint distribution of the MIs, we need to consider the case when two MIs share the same $\BX$. 

For the double-scattering case, we will resort to investigate the joint distribution of $M(z_k,\BH_k,\BP_k)$ with $\BH_{k}=\LR_k \BX_k \FS^{\frac{+}{2}}_k\BY\RT_k$, $k=1,2,...,K$, which share the same $\BX$, i.e., $\BX_{i}=\BX_{j},~i\neq j$. Note that $\BY$ is shared by all $\BH_k$. For the LBI case, we consider the channel $\BH_k=\LR_k\BX_k\RT_k$. The challenge in determining the joint distribution lies in the fact that the MIs are not independent. To the best of the authors' knowledge, there is no result regarding the joint distribution of the MIs, which will be the theoretical contribution of this paper. 

According to~(\ref{c_inf_ab}), we will derive the joint distribution of 
\begin{equation}
\label{MIsk}
(M(z_1,\BH_1,\BP_{1}),M(z_2,\BH_2,\BP_{2}),...,M(z_k,\BH_k,\BP_{K})).
\end{equation}
As there are many MIs involved, we introduce the following notation. In particular, we use the subscript to differentiate the MIs
\begin{equation}
\label{subscri}
M_{XY}=M(z_X,\BH_X,\BP_Y),
\end{equation}
where $X=\{B,E\}$ denotes different users. The index $X$ will determine the channel matrix and the noise power $\sigma^2$. $Y=\{U, W,V \}$ represents different transmit covariance matrices ($\BP_{U},\BP_W,\BP_V$). Thus, the SOP with AN can be obtained if we can determine the joint distribution of 
\begin{equation}
\begin{aligned}
(M_{BU},M_{BV},M_{EU},M_{EV}),
\end{aligned}
\end{equation}
which will be given in Section~\ref{sec_clts}. To avoid ambiguity, we will use $C$ and $D$ to denote the secrecy rate of the double-scattering and LBI channel, respectively. In the following, we will first introduce some preliminary results.

\subsection{Assumptions and Preliminary Results}
The results of this paper are developed based on the following assumptions.

\textbf{Assumption 1.} (Asymptotic Regime) $0<\lim\inf\limits_{M \ge 1}  \frac{M}{L} \le \frac{M}{L}  \le \lim \sup\limits_{M \ge 1} \frac{M}{L} <\infty$, $0<\lim \inf\limits_{M \ge 1}  \frac{M}{N_{k}} \le \frac{M}{N_{k}}  \le \lim \sup\limits_{M \ge 1}  \frac{M}{N_{k}} <\infty$.

\textbf{Assumption 2.} $\lim \sup\limits_{M\ge 1} \| \bold{R}_k\| <\infty$, $\lim \sup\limits_{M\ge 1} \| \bold{S}_{k}\| <\infty$, $\lim \sup\limits_{M\ge 1} \| \bold{T}\| <\infty$, $k=B,E$, $\lim \sup\limits_{M\ge 1} \| \BH_{T,0} \| <\infty$.

\textbf{Assumption 3.} $\inf\limits_{M\ge 1} \frac{1}{M}\Tr\bold{R}_k>0  $, $\inf\limits_{M\ge 1} \frac{1}{M}\Tr\bold{T}_k>0 $, $\inf\limits_{M\ge 1} \frac{1}{M}\Tr\bold{S}_{k}>0$.

 \textbf{A.1} is the asymptotic regime considered for the large-scale system, where the numbers of the antennas ($N_k$, $L$, and $M$) grow to infinity with  same pace. \textbf{A.2} and~\textbf{A.3} are given to exclude the extremely low-rank correlation matrices, where the rank of the correlation matrix does not increase with the number of antennas~\cite{zhang2022outage}.

We introduce the following results on the ergodic rate, which will be used to characterize the asymptotic joint distribution of the MIs.

\subsubsection{LBI case} Let $(\alpha_{k},\overline{\alpha}_{k})$ be the positive solution of the following system of equations,
 \begin{equation}
 \label{sing_eq1}
 \begin{aligned}
 \begin{cases}
  &\alpha_{k}=\frac{1}{M}\Tr\bold{R}_{k}\left(z_{k} \bold{I}_{N_{k}}+\overline{\alpha}_{k}\bold{R}_{k}   \right)^{-1},
 \\
 &\overline{\alpha}_{k}=\frac{1}{M}\Tr \bold{T}_{k}\left(\bold{I}_{L}+{\alpha}_{k}\bold{T}_{k}\right)^{-1}.
\end{cases}
  \end{aligned}
 \end{equation}
Define matrices $\bold{L}_{R,k}=\left(z_{k} \bold{I}_{N_{k}}+\overline{\alpha}_{k}\bold{R}_{k}   \right)^{-1}$ and $\bold{L}_{T,k}=\left(\bold{I}_{L}+{\alpha}_{k}\bold{T}_{k}\right)^{-1}$.
 \begin{lemma} (\cite[Theorem 1]{hachem2008new})
\label{sing_mean}
Given that assumptions~\textbf{A.1}-\textbf{A.3} hold true and $(\alpha_k,\overline{\alpha}_k)$ is the positive solution of~(\ref{sing_eq1}), the expectation of the MI $D_k=\log\det(z_k\bold{I}_{N_k}+\BH_{k}\BH^{H}_k)$ with $\BH_{k}=\FR_{k}^{\frac{1}{2}}\BX_{k}\FT_{k}^{\frac{1}{2}}$, can be approximated by
\begin{equation}
\E D_k = \overline{D}(z_k,\FR_k,\FT_k)+\mathcal{O}(\frac{1}{M}),
\end{equation}
where
\begin{equation}
\begin{aligned}
\label{dz_mean}
&\overline{D}(z_k,\FR_k,\FT_k)=\log\det(z_k\bold{I}_{N_k}+\overline{\alpha}_{k} \bold{R}_k )
\\
&
+\log\det(\bold{I}_{M}+ \alpha_k \bold{T}_k)
-M\alpha_{k}\overline{\alpha}_k.
\end{aligned}
\end{equation}
\end{lemma}

\subsubsection{Double-scattering case}
Let $(\delta_{k},\omega_{k},\overline{\omega}_{k})$ be the solution of the system of equations,
 \begin{equation}
 \label{basic_eq1}
 \begin{aligned}
 \begin{cases}
  &\delta_{k}=\frac{1}{L}\Tr\bold{R}_{k}\left(z_{k} \bold{I}_{N_{k}}+\frac{M \omega_{k}\overline{\omega}_{k}}{L\delta_{k}}\bold{R}_{k}   \right)^{-1},
 \\
 &\omega_{k}=\frac{1}{M}\Tr \bold{S}_{k}\left(\frac{1}{\delta_{k}}\bold{I}_{L}+\overline{\omega}_{k}\bold{S}_{k}\right)^{-1}, 
 \\
 &\overline{\omega}_{k}=\frac{1}{M}\Tr\bold{T}_{k}\left(\bold{I}_{M}+\omega_{k}\bold{T}_{k}\right)^{-1}.
\end{cases}
  \end{aligned}
 \end{equation}
Define $\BG_{R,k}=\left(z_{k} \bold{I}_{N_{k}}+\frac{M \omega_{k}\overline{\omega}_{k}}{L\delta_{k}}\bold{R}_{k}   \right)^{-1}$, $\BG_{S,k}=\left(\frac{1}{\delta_{k}}\bold{I}_{L}+\overline{\omega}_{k}\bold{S}_{k}\right)^{-1}$, $\BF_{S,k}=\left(\bold{I}_{L}+\delta_k \overline{\omega}_{k}\bold{S}_{k}\right)^{-1}$, and $\BG_{T,k}=\left(\bold{I}_{M}+\omega_{k}\bold{T}_{k}\right)^{-1}$.
  
\begin{lemma} (\cite[Theorem 1]{zhang2022asymptotic})
\label{dua_mean}
Given that assumptions~\textbf{A.1}-\textbf{A.3} hold true and $(\delta_k,\omega_k,\overline{\omega}_k)$ is the positive solution of~(\ref{basic_eq1}), the expectation of the MI $ C_k=\log\det(z_k\bold{I}_{N_k}+\BH_{k}\BH^{H}_k)$ with $\BH_{k}=\FR_{k}^{\frac{1}{2}}\BX_{k}\bold{S}_{k}^{\frac{+}{2}}\BY\FT_{k}^{\frac{1}{2}}$, can be approximated by
\begin{equation}
\E C_k= \overline{C}(z_k,\FR_{k},\FS_{k},\FT_{k})+\mathcal{O}(\frac{1}{M}),
\end{equation}
where
\begin{equation}
\begin{aligned}
\label{the_mean}
&
\overline{C}(z_k,\FR_{k},\FS_{k},\FT_{k})=\log\det(z_k\bold{I}_{N_k}+\frac{ M \omega_k\overline{\omega}_k}{ L\delta_k}\bold{R}_k )
\\
&
+\log\det(\bold{I}_{L}+\delta_k \overline{\omega}_k\bold{S}_k)
+\log\det(\bold{I}_{M}+\omega_k \bold{T}_k)
\\
&
-2M\omega_{k}\overline{\omega}_k.
\end{aligned}
\end{equation}
\end{lemma}
\begin{remark}
(\ref{dz_mean}) and (\ref{the_mean}) are referred to as the deterministic approximation of $D_k$ and $C_k$, respectively. We can apply the substitution $\FT'_k=\RT_k \BP_k\RT_k$ to obtain the deterministic approximation for $M(z_k,\BH_k,\BP_k)$. By Lemmas~\ref{sing_mean} and~\ref{dua_mean}, we can obtain the deterministic approximation for the expectation of any linear combination of the MIs. 
\end{remark}

\section{Asymptotic Joint MI Distribution}
\label{sec_clts}
The asymptotic distribution of a single MI has been proved to be Gaussian for both the LBI and double-scattering cases~\cite{hachem2008new, zhang2022asymptotic,zhang2022outage}. In this part, we will prove that the asymptotic joint distribution of the MIs is a joint Gaussian by setting up two CLTs. For ease of illustration, some notations are defined in Table~\ref{var_list}.
 \begin{table*}[!htbp]
\centering
\caption{Table of Notations.}
\label{var_list}
\begin{tabular}{|cc|cc|cc|cc|}
\toprule
Notations& Expression &  Notations & Expression &Notations& Expression  \\
\midrule
$\nu_{R,k,l}$ & $\frac{1}{L}\Tr\bold{R}_{k}\BG_{R,k}\FR_l\BG_{R,l}$
&
$\eta_{R,k,k,l}$& $\frac{1}{L}\Tr\FR_{k}^2\BG_{R,k}^2\FR_l \BG_{R,l}$
&
$A_{b,k}$ &  $A_{b,k,k}$
\\
$\nu_{S,k,l}$& $\frac{1}{M}\Tr\FS_k\BG_{S,k}\FS_l \BG_{S,l}$
&
$\eta_{T,k,k,l}$ & $\frac{1}{M}\Tr\bold{T}_k^2\BG_{T,k}^{2}\FT_l\BG_{T,l}$
&
 $\Delta_{S,k,l}$  & $1-\nu_{S,k,l}\nu_{T,k,l}$
\\
$\nu_{S,I,k,l}$& $\frac{1}{M}\Tr\FSr_k\FSl_l\BG_{S,l}\BG_{S,k}$
&
 $\eta_{S,I,k,k,l}$ &  $\frac{1}{M}\Tr\FS_{k}^2\BG_{S,k}^{2}\FS_l\BG_{S,l}$
 &
$\Delta_{k,l}$ & $1-\frac{M\overline{\omega}_{k}\overline{\omega}_{l}\nu_{R,k,l}\nu_{S,k,l}}{L\delta_{k}\delta_{l}}-\frac{M\nu_{R,k,l}\nu_{S,I,k,l}^2\nu_{T,k,l}}{L\delta_{k}^2\delta_{l}^2\Delta_{S,k,l}}$
  \\
  $\nu_{T,k,l}$& $\frac{1}{M}\Tr\FT_k\BG_{T,k}\FT_l\BG_{T,l}$
&
$\eta_{S,k,k,l}$  & $\frac{1}{M}\Tr\FS_{k}^2\BG_{S,k}^{2}\FS_l\BG_{S,l}$ 
  &
  $  \Xi_{k,l}$ &$ 1-\gamma_{R,k,l}\gamma_{T,k,l}$
\\
  $  \gamma_{R,k,l}$ &$ \frac{1}{M}\Tr\FR_k \bold{L}_{R,k} \FR_l \bold{L}_{R,l}$
&
 $  \gamma_{T,k,l}$ &$ \frac{1}{M}\Tr\FT_k \bold{L}_{T,k} \FT_l \bold{L}_{T,l}$
&
&
 \\
\bottomrule
\end{tabular}
\end{table*}

\label{clt_the}
\begin{theorem} 
\label{CLT_MIS}
(The asymptotic joint distribution of the MIs for double-scattering channel) Given assumptions \textbf{A.1}-\textbf{A.3} and a sequence of $K$ MIs, i.e., $C_{k}=\log\det(z_k \bold{I}_{N_{k}}+\BH_{k}\BH^{H}_{k})$, for $k=1,,2,...,K$ with $\BH_{k}=\FR_{k}^{\frac{1}{2}}\BX_{k}\bold{S}_{k}^{\frac{+}{2}}\BY\FT_{k}^{\frac{1}{2}}$, there holds true that
\label{CLT_MI}
\begin{equation}
\begin{aligned}
(C_1-\overline{C}_{1},...,C_{K}-\overline{C}_{K})
 \xrightarrow[N \rightarrow \infty]{\mathcal{D}}  \mathcal{N}(\bold{0}_{K},\bold{M}).
\end{aligned}
\end{equation}
Here, $\overline{C}_{k}$ is determined by
\begin{equation}
\overline{C}_{k}=\overline{C}(z_k,\FR_k,\FS_k,\FT_k),
\end{equation}
which is given in~(\ref{the_mean}). The $(i,j)$-th entry of $\BM$, which represents the asymptotic covariance between $C_i$ and $C_j$, can be expressed as 
\begin{equation}
\label{doub_cor}
[\BM]_{i,j}=-\mathbbm{1}_{\{\BX_{i}=\BX_{j}\}}\log(\Delta_{i,j})-\log(\Delta_{S,i,j})
\end{equation}
where $\Delta_{i,j}$ and $\Delta_{S,i,j}$ are given in Table~\ref{var_list}.
\end{theorem}
\begin{proof}
The proof of Theorem~\ref{CLT_MI} is given in Appendix~\ref{dua_clt_proof}.
\end{proof}

\begin{theorem} 
\label{single_clt}
(The asymptotic joint distribution of the MIs for LBI channel) Given assumptions \textbf{A.1}-\textbf{A.3} and a sequence of $K$ MIs, i.e., $D_{k}=\log\det(z_k \bold{I}_{N_{k}}+\BH_{k}\BH^{H}_{k})$, for $k=1,2,...,K$ with $\BH_{k}=\FR_{k}^{\frac{1}{2}}\BX_{k}\FT_{k}^{\frac{1}{2}}$, there holds true that
\begin{equation}
\begin{aligned}
(D_{1}-\overline{D}_{1},...,D_{K}-\overline{D}_{K})
 \xrightarrow[N \rightarrow \infty]{\mathcal{D}}  \mathcal{N}(\bold{0}_{K},\BF).
\end{aligned}
\end{equation}
Here, $\overline{D}_{k}$ is determined by
\begin{equation}
\overline{D}_k=\overline{D}(z_k,\FR_k,\FT_k),
\end{equation}
which is given in~(\ref{dz_mean}). The $(i,j)$-th entry of $\BF$, which represents the asymptotic covariance between $D_i$ and $D_j$, can be expressed as 
\begin{equation}
\label{sing_F}
[\BF]_{i,j}=-\mathbbm{1}_{\{\BX_{i}=\BX_{j}\}}\log(\Xi_{i,j}),
\end{equation}
where $\Xi_{i,j}$ is given in Table~\ref{var_list}.
\end{theorem}
\begin{proof} The proof is simpler than that of Theorem~\ref{CLT_MIS} and the computation process can be found in~\cite{hachem2008new}. We omit the proof here.
\end{proof}
\begin{remark} Theorem~\ref{CLT_MIS} indicates that the asymptotic joint distribution of the MIs is a Gaussian distribution. $[\BM]_{i,i}$ denotes the variance of $C_{i}$ and $[\BM]_{i,j},~i\neq j$ represents the covariance between $C_i$ and $C_{j}$. When $K=1$, the CLT is equivalent to the results in~\cite[Theorem~2]{zhang2022asymptotic} and~\cite[Theorem~2]{zhang2022outage}. In Theorem~\ref{CLT_MIS}, all the MIs are correlated since they share the same $\BY$ and the entries of the covariance matrix can be categorized into two types: 1. $\BH_i$ and $\BH_j$ share the same $\BX$, which results in a larger covariance; 2. $\BH_i$ and $\BH_j$ have independent $\BX$, which leads to a smaller covariance. In Theorem~\ref{single_clt}, there is only one form of non-zero covariance, i.e., when $\BX_i$ and $\BX_j$ are identical.
\end{remark}

\section{Ergodic Secrecy Rate}
\label{sec_esr}
In this section, we will give the closed-form approximations for the ESR of IRS-aided MIMO channels.
\begin{theorem} (ESR of Wiretap Systems)
\label{mean_NA}
The secrecy rate of the MIMOME wiretap system, over double-scattering and LBI channels, can be evaluated by
\begin{equation}
\begin{aligned}
{C}_{S}(\BP_W)=\overline{C}_{S}(\BP_W)+\BO(\frac{1}{M}),
\\
{D}_{S}(\BP_W)=\overline{D}_{S}(\BP_W)+\BO(\frac{1}{M}),
\end{aligned}
\end{equation}
respectively, where  
\begin{equation}
\label{esr_wire}
\begin{aligned}
&\overline{C}_{S}(\BP_W)=\lceil
\overline{C}(\sigma^2_{B},\FR_{B},\FS_{B},\RT\BP_{W}\RT)
\\
&
-\overline{C}(\sigma^2_{E},\FR_{E},\FS_{E},\RT\BP_W\RT)
\\
&-N_{B}\log(\sigma^2_{B})
+N_{E}\log(\sigma^2_{E}) \rceil^{+},
\\
&\overline{D}_{S}(\BP_W)=
\lceil 
\overline{D}(\sigma^2_{B},\FR_{B}, \FT^{\frac{+}{2}}_{B}\BP_W\FT^{\frac{-}{2}}_{B})
\\
&
-\overline{D}(\sigma^2_{E},\FR_{E},\FT^{\frac{+}{2}}_{E}\BP_W\FT^{\frac{-}{2}}_{E})
\\
&
-N_{B}\log(\sigma^2_{B})
+N_{E}\log(\sigma^2_{E})
 \rceil^{+}.
\end{aligned}
\end{equation}
\end{theorem}
The ESR of the AN-aided system is give by the following theorem.
\begin{theorem} (ESR of AN-aided Systems)
\label{mean_AN}
The secrecy rate of IRS-aided MIMO systems with AN, over double scattering and LBI channels, can be evaluated by
\begin{equation}
\begin{aligned}
\label{esr_an}
{C}_{AN}(\BP_W,\BP_V)=\overline{C}_{AN}(\BP_W,\BP_V)+\BO(\frac{1}{M}),
\\
{D}_{AN}(\BP_W,\BP_V)=\overline{D}_{AN}(\BP_W,\BP_V)+\BO(\frac{1}{M}),
\end{aligned}
\end{equation}
where 
\begin{equation}
\begin{aligned}
\label{an_esr_exp}
&\overline{C}_{AN}(\BP_W,\BP_V)=\lceil
\overline{C}(\sigma^2_{B},\FR_{B},\FS_{B},\RT\BP_{U}\RT)
\\
&-\overline{C}(\sigma^2_{B},\FR_{B},\FS_{B},\RT\BP_{V}\RT)
\\
&
-\overline{C}(\sigma^2_{E},\FR_{E},\FS_{E},\RT\BP_{U}\RT)
\\
&
+\overline{C}(\sigma^2_{E},\FR_{E},\FS_{E},\RT\BP_{V}\RT) \rceil^{+},
\\
&\overline{D}_{AN}(\BP_W,\BP_V)=\lceil
\overline{D}(\sigma^2_{B},\FR_{B},\FT^{\frac{+}{2}}_{B}\BP_U\FT^{\frac{-}{2}}_{B})
\\
&
-\overline{D}(\sigma^2_{B},\FR_{B},\FT^{\frac{+}{2}}_{B}\BP_V\FT^{\frac{-}{2}}_{B})
-\overline{D}(\sigma^2_{E},\FR_{E},\FT^{\frac{+}{2}}_{E}\BP_U\FT^{\frac{-}{2}}_{E})
\\
&
+\overline{D}(\sigma^2_{E},\FR_{E},\FT^{\frac{+}{2}}_{E}\BP_V\FT^{\frac{-}{2}}_{E})
 \rceil^{+}.
\end{aligned}
\end{equation}
\end{theorem}
\begin{proof} We only need to replace the four MIs in~(\ref{ESR_AN}) by their deterministic approximations according to Lemmas~\ref{sing_mean} and~\ref{dua_mean}.
\end{proof}
\begin{remark} The wiretap model is a special case of the AN-aided model as Theorem~\ref{mean_NA} is equivalent to Theorem~\ref{mean_AN} when $\BP_{V}=\bold{0}$.
\end{remark}

\section{Secrecy Outage Probability}
\label{sec_sop}
In this section, we will give the closed-form expressions for the SOPs based on Theorems~\ref{CLT_MIS} and~\ref{single_clt} in Section~\ref{clt_the}.
\begin{proposition} 
\label{s_sop}
The SOP of the IRS-aided wiretap MIMO system can be approximated by
\begin{equation}
P_{S,out}(R)\approx\Phi(\frac{R-\overline{R}_{S}}{\sqrt{V_S}}),
\end{equation}
where
\begin{equation}
\label{VS_exp}
V_{S}=
\begin{cases}
&-\log(\Delta_{BW})-\log(\Delta_{EW})
\\
&
+2\log(\Delta_{BW,EW}),~\text{double-scattering},
\\
&-\log(\Xi_{BW})-\log(\Xi_{EW}),~\text{LBI}.
\end{cases}
\end{equation}
The ESR $\overline{R}$ of the double-scattering and LBI chanel are given in~(\ref{esr_wire}). The definition of the subscript in~(\ref{VS_exp}) is given in~(\ref{subscri}).
\end{proposition}
\begin{proof}Here we only prove the double-scattering case since the LBI case can be proved similarly. Given the Gaussianity proved by Theorem 1 and the ESR in~(\ref{esr_wire}), we only need to determine the variance of the secrecy rate. According to~(\ref{C_inf1}), we need to determine the variance of $C_{BW}-C_{EW}$. For that purpose, we first derive the distribution of $C_{BW}-C_{EW}$. By Theorem~\ref{CLT_MI}, we know that the joint distribution of $(C_{BW},C_{EW})$ converges to a joint Gaussian distribution, whose covariance matrix is given by
\begin{equation}
\BM=
\begin{bmatrix}
-\log(\Delta_{BW,BW})     & -\log(\Delta_{S,BW,EW}) \\
-\log(\Delta_{S,BW,EW}) & -\log(\Delta_{EW,EW})
\end{bmatrix}.
\end{equation}
Therefore, the variance of $C_{B}-C_{E}$ is given by $\bold{u}^{T}\BM\bold{u}$, where $\bold{u}=(1,-1)^{T}$.
\end{proof}

\begin{proposition} 
\label{an_sop}
The SOP of the AN-aided system over double-scattering channel can be approximated by
\begin{equation}
P_{AN,out}(R)\approx \Phi(\frac{R-\overline{R}_{AN}}{\sqrt{V_{AN}}}),
\end{equation}
where 
\begin{equation}
V_{AN}=\bold{u}^{T}\BM_{AN}\bold{u},
\end{equation}
 with $\bold{u}=(1,-1,-1,1)^{T}$. Here the covariance matrix between $C_{BU}$, $C_{BV}$, $C_{EU}$, and $C_{EV}$ is given by 
\begin{equation}
\begin{aligned}
\BM_{AN}&=
\begin{bmatrix}
\BC_{B,B}     & \BC_{B,E}  \\
\BC_{E,B}     & \BC_{E,E}
\end{bmatrix},
\end{aligned}
\end{equation}
where each sub-matrix in $\BM_{AN}$ can be determined as  
\begin{equation}
\label{BC_def}
\BC_{i,j}=
\begin{bmatrix}
[\BM]_{iU,iU}     & [\BM]_{iU,jV}  \\
[\BM]_{jV,iU}& [\BM]_{jV,jV}
\end{bmatrix},
~i,j=B,E,
\end{equation}
and $[\BM]_{i,j}$ is given in~(\ref{doub_cor}). The LBI case can be obtained by replacing $[\BM]_{i,j}$ with $[\BF]_{i,j}$ in~(\ref{sing_F}). The ESR $\overline{R}$ of the double-scattering and LBI cases are given in~(\ref{esr_an}). \end{proposition} 
\begin{proof}  The proof is similar to that of Proposition~\ref{s_sop}, so we omit it here.
\end{proof}

\subsection{Multiple Multi-antenna Eavesdroppers}
We can extend the results in Propositions~\ref{s_sop} and~\ref{an_sop} to the scenario when there are multiple MEs. The outage probability with $K$ MEs is given as~\cite{multipleoutage}
\begin{equation}
\begin{aligned}
&P_{out}(R)=\Prob(C_{B}(\BP_W,\BP_V)
\\
&-\max_{k=1,2,...,K}C_{E}(\BP_W,\BP_V)<R)\approx F_{K}(R),
\end{aligned}
\end{equation}
where
\begin{equation}
\begin{aligned}
&F_{K}(R)=1-\int_{R}^{\infty} \mathrm{d}x_1\int_{R}^{\infty} \mathrm{d}x_2...\int_{R}^{\infty} (2\pi)^{-\frac{K}{2}}\times
\\
&
[\det(\BQ_{sc})]^{-\frac{1}{2}}
 \exp\left(-\frac{1}{2}(\bold{x}-\boldsymbol{\mu}_{sc})^{T}\BQ^{-1}_{sc} (\bold{x}-\boldsymbol{\mu}_{sc})   \right)     \mathrm{d}x_K,
\end{aligned}
\end{equation}
and the subscript $sc=AN,S$ denotes the schemes with and without AN, respectively. For the AN-aided case, $[\BQ_{AN}]_{i,j}=\bold{u}^{T}\BC_{AN,i,j}\bold{v}$, $\bold{u}=\left(1,-1,-1,1,0,0 \right)^{T}$, and $\bold{v}=\left(1,-1,0,0,-1,1 \right)^{T}$, with $i,j=1,2,...,K$. 
 For the double-scattering channel with AN, the mean and covariance can be given by
\begin{equation}
\begin{aligned}
\boldsymbol{\mu}_{AN}&=(\overline{C}_{AN,1},\overline{C}_{AN,2},...,\overline{C}_{AN,K})^{T},
\\
\BC_{AN,i,j}&=
\begin{bmatrix}
\BC_{B,B}     & \BC_{B,E_i} & \BC_{B,E_j} \\
\BC_{E_i,B}     & \BC_{E_i,E_i} & \BC_{E_i,E_j} \\
\BC_{E_j,B}     & \BC_{E_j,E_i} & \BC_{E_j,E_j}
\end{bmatrix},
\end{aligned}
\end{equation}
respectively. $\BC_{i,j}$ is given in~(\ref{BC_def}). For the double-scatter channel without AN, $[\BQ_S]_{i,j}=\bold{u}^{T}\BC_{S,i,j}\bold{v}$ with
 $\bold{u}=\left(1,-1,0 \right)^{T}$, $\bold{v}=\left(1,0,-1 \right)^{T}$. For this case, the mean and variance are given by 
\begin{equation}
\begin{aligned}
\boldsymbol{\mu}_{S}&=(\overline{C}_{S,1},\overline{C}_{S,2},...,\overline{C}_{S,K})^{T},
\\
\BC_{S,i,j}&\!=\!\!\!
\begin{bmatrix}
\!\![\BM]_{BW,BW}   \!  & \! \! [\BM]_{BW,E_i W } \!& \!\![\BM]_{BW,E_j W} \\
\![\BM]_{E_i W,BW}   \!  &\! [\BM]_{E_i W,E_i W}\! & \![\BM]_{E_i  W,E_j W} \!\\
[\BM]_{E_j W,BW}    \! & \![\BM]_{E_j W,E_i W} \!& \![\BM]_{E_j W,E_j W}\!
\end{bmatrix},
\end{aligned}
\end{equation}
where $[\BM]_{i,j}$ is given in~(\ref{doub_cor}). The results for the LBI case can be derived similarly.

\section{Optimization Based on Statistical CSI}
\label{sec_alg}
In this section, we will present two algorithms to optimize the ESR and the SOP of IRS-aided MIMO systems, respectively. In the following analysis, we will ignore the ceiling function in ESR.
\subsection{ESR Optimization by Jointly Designing the Covariance Matrices and Phase Shifts}
The ESR maximization problem can be formulated as
 \begin{equation}
  \begin{aligned}
\mathcal{P}1:~&\max_{\bold{\Theta},\BP_W,\BP_V} \overline{D}_{AN}(\BP_W,\BP_V,\bold{\Theta}),~s.t.
\\
&
\mathcal{C}_1 :\Tr\BP_W+\Tr\BP_V \le M P.
\\
&
\mathcal{C}_2 :\BP_W \succeq \bold{0},~~\BP_V \succeq \bold{0}
\\
&\mathcal{C}_3: \bold{\Theta}=\diag(e^{\jmath \theta_{1}},e^{\jmath \theta_{1}},...,e^{\jmath \theta_{L}}).
 \end{aligned}
 \end{equation}
which is non-convex due to the unimodular constraints on the phase shifts and the non-concavity of the objective function. To overcome the difficulty, an AO algorithm will be considered, which results in the following two sub-problems, i.e., $\mathcal{P}2$ and $\mathcal{P}3$.
 \begin{equation}
  \begin{aligned}
\mathcal{P}2:~&\max_{\BP_{W},\BP_{V}} \overline{D}_{AN}(\BP_{W},\BP_{V}),~s.t.~\mathcal{C}_1,~\mathcal{C}_2.
 \end{aligned}
 \end{equation}
 and
  \begin{equation}
  \begin{aligned}
\mathcal{P}3:~&\max_{\bold{\Theta}} \overline{D}_{AN}(\bold{\Theta}),~s.t.~\mathcal{C}_3.
 \end{aligned}
 \end{equation}
 Next, we will solve $\mathcal{P}2$ and $\mathcal{P}3$, respectively.
\subsubsection{Optimization of the covariance matrices for the signal and AN}
\label{cov_opt}
Note that the objective function of $\mathcal{P}2$ is not concave due to the two negative $\log$ terms in~(\ref{an_esr_exp}). To utilize the classical convex optimization algorithms, we apply the successive convex approximation (SCA) approach to handle the negative $\log$ terms, which is denoted as $N(\BP_{W},\BP_{V})$, i.e.,
  \begin{equation}
  \begin{aligned}
 & N(\BP_{W},\BP_{V})=-\overline{D}(\sigma^2_{E},\FR_{E},\FT^{\frac{+}{2}}_{E}\BP_U\FT^{\frac{-}{2}}_{E})
  \\
  &-\overline{D}(\sigma^2_{B},\FR_{B},\FT^{\frac{+}{2}}_{B}\BP_V\FT^{\frac{-}{2}}_{B}).
  \end{aligned}
 \end{equation}
Specifically, we aim to obtain a convex upper bound for the objective function by an iterative approach. To facilitate the SCA, we construct a global under-estimator for the negative terms. We use the superscript $(t)$ as the iteration index. It has been show in~\cite{dumont2010capacity} that for any given positive semi-definite matrix $\BQ$, the function $\overline{D}(\sigma^2_k,\FR_k,\FT^{\frac{+}{2}}_{k}\BQ\FT^{\frac{-}{2}}_{k})$ is strictly concave with respect to $\BQ$. Therefore, given $t$, $N(\BP_{W},\BP_{V})$ is lower bounded by its first-order Taylor expansion~\cite{yu2020robust},

\begin{equation}
\begin{aligned}
\label{taylor_sca}
& N(\BP_{W},\BP_{V})\ge N(\BP_{V}^{(t)},\BP_{W}^{(t)})
\\
&
- \Tr\nabla_{\BP_{W}}N(\BP_{W}^{(t)},\BP_{V}^{(t)})(\BP_{W}-\BP_{W}^{(t)})
\\
&
-\Tr \nabla_{\BP_{V}}N(\BP_{W}^{(t)},\BP_{V}^{(t)})(\BP_{V}-\BP_{V}^{(t)}),
\end{aligned}
\end{equation}
where 
\begin{equation}
\begin{aligned}
 &\nabla_{\BP_{W}}N(\BP_{W}^{(t)},\BP_{V}^{(t)})=\alpha_{E}^{(t)}\FT_{E}^{{\frac{-}{2}} } \bold{L}_{T,EU}\FT_{E}^{\frac{+}{2}} 
 \\
  &\nabla_{\BP_{V}}N(\BP_{W}^{(t)},\BP_{V}^{(t)})=\alpha_{E}^{(t)}\FT_E^{{\frac{-}{2}} } \bold{L}_{T,EU}\FT_E^{{\frac{+}{2}} }
  \\
  &
  +\alpha_{B}^{(t)}\FT_B^{{\frac{-}{2}} }  \bold{L}_{T,BV}\FT_B^{{\frac{+}{2}} } .
\end{aligned}
\end{equation}
By employing this lower bound for the objective function, $\mathcal{P}2$ can be reformulated as the following convex optimization problem,
\begin{equation}
\begin{aligned}
\label{SCA_opt}
&\mathcal{P}4:~\max_{\BP_{W},\BP_{V}} \overline{D}_{AN,SCA}(\BP_W,\BP_V|\BP_{W}^{(t)},\BP_{V}^{(t)})
\\
&
:= \overline{D}(\sigma^2_{B},\FR_{B}, \FT^{\frac{+}{2}}_{B}\BP_{U}\FT^{\frac{-}{2}}_{B})
\\
&
+\overline{D}(\sigma^2_{E},\FR_{E},\FT^{\frac{+}{2}}_{E}\BP_{U}\FT^{\frac{-}{2}}_{E}) 
\\
&
- \Tr\nabla_{\BP_{W}}N(\BP_{W}^{(t)},\BP_{V}^{(t)})\BP_{W}
\\
&
-\Tr \nabla_{\BP_{V}}N(\BP_{W}^{(t)},\BP_{V}^{(t)})\BP_{V}+const,
~s.t.~\mathcal{C}_1,~\mathcal{C}_2.
\end{aligned}
\end{equation}
The solution of $\mathcal{P}4$ can not be directly obtained since $\alpha$ and $\overline{\alpha}$ in $\overline{D}$ are coupled with $\BP_{W}$ and $\BP_{V}$. $f(\overline{\alpha}_{B,U})=\overline{D}(\sigma^2_{B},\FR_{B},\FT^{\frac{+}{2}}_{B}\BP_{U}\FT^{\frac{-}{2}}_{B})$ can be regarded as a function with respect to $\overline{\alpha}_{B,U}$, where $\alpha_{B,U}=\frac{1}{M}\Tr\FR_{B}\left(\sigma^2_B\bold{I}_{N}+\overline{\alpha}_{B,U}\FR_{B} \right)^{-1}$ if we fix other parameters. Furthermore, we can obtain that $f(\overline{\alpha})$ is quasi-concave with respect to $\overline{\alpha}$ because
\begin{equation}
\begin{aligned}
&\frac{\mathrm{d} f(\overline{\alpha})}{\mathrm{d} \overline{\alpha}}=
-(\Tr\FT^{\frac{+}{2}}_{B}\BP_{U}\FT^{\frac{-}{2}}_{B}(\bold{I}_{M}+{\alpha}\FT^{\frac{+}{2}}_{B}\BP_{U}\FT^{\frac{-}{2}}_{B})^{-1} 
\\
&
-M\overline{\alpha})
\frac{\Tr\FR_{B}^2(\sigma_B^2\bold{I}_{N_B}+\overline{\alpha}\FR_{B} )^{-2}}{M}
\end{aligned}
\end{equation}
is larger than $0$ when $\Tr\FT^{\frac{+}{2}}_{B}\BP_{U}\FT^{\frac{-}{2}}_{B}(\bold{I}_{M}+{\alpha}\FT^{\frac{+}{2}}_{B}\BP_{U}\FT^{\frac{-}{2}}_{B})^{-1}<M\overline{\alpha}$ and less than $0$ when $\Tr\FT^{\frac{+}{2}}_{B}\BP_{U}\FT^{\frac{-}{2}}_{B}(\bold{I}_{M}+{\alpha}\FT^{\frac{+}{2}}_{B}\BP_{U}\FT^{\frac{-}{2}}_{B})^{-1}>M\overline{\alpha}$. The solution of $\Tr\FT^{\frac{+}{2}}_{B}\BP_{U}\FT^{\frac{-}{2}}_{B}(\bold{I}_{M}+{\alpha}\FT^{\frac{+}{2}}_{B}\BP_{U}\FT^{\frac{-}{2}}_{B})^{-1}=M\overline{\alpha}$, denoted as $\overline{\alpha}^{*}$, corresponds to the minimum of $f(\overline{\alpha})$. Similar analysis can be performed on $\overline{D}(\sigma^2_{E},\FR_{E},\FT^{\frac{+}{2}}_{E}\BP_V\FT^{\frac{-}{2}}_{E})$. Therefore, by omitting unrelated variables, we rewrite $\mathcal{P}4$ as
\begin{equation}
\begin{aligned}
\mathcal{P}5:~&\max_{\BP_{W},\BP_{V}} \min_{ \gamma_1>0,\gamma_2>0} F(\gamma_1,\gamma_2,\BP_{W},\BP_{V})
:=
\\
&
 f_1(\gamma_1,\BP_W,\BP_V)+f_2(\gamma_2,\BP_V)
\\
&
- \Tr\nabla_{\BP_{W}}N(\BP_{W}^{(t)},\BP_{V}^{(t)})\BP_{W}
\\
&
-\Tr \nabla_{\BP_{V}}N(\BP_{W}^{(t)},\BP_{V}^{(t)})\BP_{V},
~s.t.~\mathcal{C}_1,~\mathcal{C}_2.
\end{aligned}
\end{equation}
Note that given any matrix $\BA$ and positive semi-definite matrix $\BQ$, $\log(\det(\bold{I}+\BA\BQ\BA^{H}))$ is strictly concave with respect to $\BQ$. By the analysis above, we can conclude that $F(\gamma_1,\gamma_2,\BP_{W},\BP_{V})$ is concave with respect to $\BP_{W}$ and $\BP_{V}$, and quasi-convex with respect to $\gamma_1$ and $\gamma_2$. According to the generalized mini-max theorem~\cite{sion1958general}, a saddle-point exists. For any function $f$, we have $\max_{\gamma_{1}}\min_{\gamma_{2}}   f(\gamma_{1},\gamma_2)\le \min_{\gamma_{2}}\max_{\gamma_{1}} f(\gamma_{1},\gamma_2)$ and the equality holds if the saddle point exists~\cite{boyd2004convex}. By the analysis in~\cite{wen2006asymptotic}, the optimization problem $\mathcal{P}5$  can be rewritten as
\begin{equation}
\begin{aligned}
 \min_{ \gamma_1>0,\gamma_2>0} & \max_{\BP_{W},\BP_{V}}F(\gamma_1,\gamma_2,\BP_{W},\BP_{V})
 \\
 &
- \Tr\nabla_{\BP_{W}}N(\BP_{W}^{(t)},\BP_{V}^{(t)})\BP_{W}
\\
&
-\Tr \nabla_{\BP_{V}}N(\BP_{W}^{(t)},\BP_{V}^{(t)})\BP_{V},
~s.t.~\mathcal{C}_1,~\mathcal{C}_2.
\end{aligned}
\end{equation}
The inner problem is the maximization of a concave function~\cite{boyd2004convex}, which can be written as
\begin{equation}
\begin{aligned}
\mathcal{P}6:~& \max_{\BP_{W},\BP_{V}}
\log\det(\bold{I}_{M}+\alpha_{BU}(\gamma_{1})\FT^{\frac{+}{2}}_{B}\BP_{U}\FT^{\frac{-}{2}}_{B})
\\
&
+\log\det(\bold{I}_{M}+\alpha_{EU}(\gamma_{2})\FT^{\frac{+}{2}}_{E}\BP_{V}\FT^{\frac{-}{2}}_{E})
\\
&
- \Tr\nabla_{\BP_{W}}N(\BP_{W}^{(t)},\BP_{V}^{(t)})\BP_{W}
\\
&
-\Tr \nabla_{\BP_{V}}N(\BP_{W}^{(t)},\BP_{V}^{(t)})\BP_{V},
~s.t.~\mathcal{C}_1,~\mathcal{C}_2,
\end{aligned}
\end{equation}
and resolved by CVX~\cite{grant2008cvx}. $\mathcal{P}5$ can be resolved by an iterative approach given in Algorithm~\ref{opt_PWPV}.
\begin{algorithm} 
\caption{Iterative Algorithm of Solving $\mathcal{P}5$} 
\label{opt_PWPV} 
\begin{algorithmic}[1] 
\REQUIRE  $\BP_W^{(0)}$, $\BP_V^{(0)}$,  $\nabla_{\BP_{W}}N(\BP_{W}^{(t)},\BP_{V}^{(t)})$, and $\nabla_{\BP_{V}}D(\BP_{W}^{(t)},\BP_{V}^{(t)})$. Set $c=0$.
\REPEAT
\STATE Compute $\gamma_{1}^{(c)}$ and $\gamma_{2}^{(c)}$ by solving the systems of equations~(\ref{sing_eq1}) based on $\BP_W^{(c)}$, $\BP_V^{(c)}$.
\STATE Obtain $\BP_W^{(c+1)}$, $\BP_V^{(c+1)}$ by solving $\mathcal{P}6$ using CVX.
\STATE $c \leftarrow  c+1$.
\UNTIL Convergence.
\ENSURE  $\BP_W^{(c)}$, $\BP_V^{(c)}$.
\end{algorithmic}
\end{algorithm}

\subsubsection{Optimization of the phase shifts}
\label{phase_opt}
In this section, we will present an algorithm to optimize the phase shifts given the transmit covariance matrices. Due to the non-convexity of the unimodular constraints and complex relations of parameters induced by the system of equations, we adopt a gradient ascent approach to determine the optimal phase shifts, which has been widely used in the design based on statistical CSI~\cite{zhang2022outage,liu2021secrecy,zhang2021large}.

We can find a  suboptimal solution by resorting to the gradient method. By similar computations in~\cite{zhang2022large}, we can obtain the derivatives of $\overline{D}$ with respect to $\theta_{i}$,
\begin{equation}
\begin{aligned}
\label{cs_der}
&\frac{\partial \overline{D}(\sigma^2_{X},\FR_{X},\RT_{S,X}\boldsymbol{\Theta}\BH_{T,0}\BP_Y\BH_{T,0}^{H}\boldsymbol{\Theta}^{H}\RT_{S,X} )}{\partial \theta_{i}}
\\
&
=\alpha_{XY}\mathrm{Tr} [\RT_{S,X}[(\BH_{T,0}\BP_Y\BH_{T,0}^{H}) \otimes \bold{F} ] \RT_{S,X}\times
\\
&
\left( \bold{I} + \alpha_{XY} \RT_{S,X}\boldsymbol{\Theta}\BH_{T,0}\BP_Y\BH_{T,0}^{H}\boldsymbol{\Theta}^{H} \RT_{S,X}\right)^{-1}],
\end{aligned}
\end{equation}
where $\alpha_{XY}$ is the solution of the system equations~(\ref{sing_eq1}) with $X=B,E$ and $Y=U,V$. In~(\ref{cs_der}), $\bold{F}\in \mathbb{C}^{L\times L}$ is given by
\begin{equation}
\left[\bold{F}\right]_{p,q}=\left\{
\begin{aligned}
& \jmath e^{\jmath (\theta_{i}-\theta_{q})} ,  p = i, \\
& -\jmath e^{\jmath (\theta_{p}-\theta_{i})} ,  q = i, \\
&0,  \text{otherwise}. \\
\end{aligned}
\right.
\end{equation}
For ease of illustration, we introduce the following notation
\begin{equation}
\begin{aligned}
K({\boldsymbol{\theta}})=\overline{D}_{AN}(\bold{\Theta}).
\end{aligned}
\end{equation}
The gradient $\frac{\partial K(\boldsymbol{\theta})}{\partial \theta_i}$ can be obtained by the linearity of derivatives. The backtrack line search method~\cite{boyd2004convex} is adopted to find the step size $\gamma$ such that
\begin{equation}
\label{grad_up}
K\left(\bm{\theta}+\gamma\nabla K\left(\bm{\theta} \right) \right) \ge  K(\bm{\theta})+
c\gamma \|\nabla K(\bm{\theta})\|,
\end{equation}
where $\nabla K(\bm{\theta})=\left(\frac{K(\bm{\theta})}{\partial \theta_{1}},\frac{K(\bm{\theta})}{\partial \theta_{2}},...,\frac{K(\bm{\theta})}{\partial \theta_{{L}}}\right)^{T}$ and $0<c<1$ is a constant. 
Next, we provide the gradient ascent method. We use the Armijo-Goldstein (AG) line search method~\cite{armijo1966minimization} to find an expected increase of the objective function based on the local gradients~\cite{kammoun2020asymptotic,liu2021secrecy,zhang2022large}. Please note that superscript $(c)$ and $(t)$ represent the indices of iterations in Algorithm~\ref{opt_PWPV} and Algorithm~\ref{AO_trans}, respectively. If we discard $\BP_{V}$, Algorithm~\ref{AO_trans} is also applicable to the joint optimization of the wiretap systems.


\subsubsection{AO algorithm}
According to the analysis in Sections~\ref{cov_opt} and~\ref{phase_opt}, the overall AO algorithm is given in Algorithm~\ref{AO_trans}. 
\begin{algorithm} 
\ref{AO_trans}
\caption{AO Algorithm of Transmit Covariance Matrices $\BP_{W}$, $\BP_{V}$ and Phase Shift Matrix $\bold{\Theta}$} 
\label{AO_trans}
\begin{algorithmic}[1] 
\REQUIRE  $\bm{\theta}^{\left(0 \right)}, \BP_{W}^{\left(0 \right)}, \BP_{V}^{\left(0 \right)}$, and set $t=0$.
\REPEAT
\STATE Compute the convex approximation by~(\ref{SCA_opt}). 
\STATE Obtain $\BP_{W}^{\left(t+1 \right)}$ and $\BP_{V}^{\left(t+1 \right)}$ with fixed $\bm{\theta}^{(t)}$ by utilizing Algorithm~\ref{opt_PWPV}.
	\STATE Compute $\nabla K\left(\bm{\theta}^{(t)} \right)$ by~(\ref{cs_der}).
	\STATE Find the step size $\gamma$ by backtrack line search~\cite{boyd2004convex}.
	\STATE $\bm{\theta}^{(t+1)}= \bm{\bold{\theta}}^{(t)}+\gamma  \nabla K\left(\bm{\theta}^{(t)} \right)$.
\STATE $t \leftarrow  t+1$
\UNTIL Convergence.
\ENSURE  $\bm{\theta}^{(t)}, \BP_{W}^{(t)},\BP_{V}^{(t)}$.
\end{algorithmic}
\end{algorithm}


\subsubsection{The convergence of Algorithm~\ref{AO_trans}}: Here we first show that the iterative design of the covariance matrices is non-decreasing. For the case the maximum of the SCA problem in $\mathcal{P}4$ is $\overline{D}_{AN,SCA}(\BP_W^{*},\BP_V^{*}|\BP_W^{(t)},\BP_V^{(t)})$, which can be solved by Algorithm~\ref{opt_PWPV}, we have
\begin{equation}
\begin{aligned}
&\overline{D}_{AN}(\BP_{W}^{*},\BP_{V}^{*}) \ge \overline{D}_{AN,SCA}(\BP_W^{*},\BP_V^{*}|\BP_W^{(t)},\BP_V^{(t)})  \ge
\\
& \overline{D}_{AN,SCA}(\BP_W^{(t)},\BP_V^{(t)} |\BP_W^{(t)},\BP_V^{(t)}) 
= \overline{D}_{AN}(\BP_{W}^{(t)},\BP_{V}^{(t)}),
\end{aligned}
\end{equation}
where the first inequality follows from the concavity on $\BP_W$ and $\BP_V$ in~(\ref{taylor_sca}). The optimization of phase shifts is obviously non-decreasing. Thus, in each subproblem, we can obtain a non-decreasing value of the objective function. As a result, the algorithm will converge to a stationary point.


\subsection{Optimization of SOP}
\label{opt_sop}

Given statistical CSI, the optimization problem for the SOP with respect to the phase shifts can be formulated as
 \begin{equation}
  \begin{aligned}
\mathcal{P}8:~&\min_{\bold{\Theta}} P_{S,out}(R),~s.t.~ \mathcal{C}_3.
 \end{aligned}
 \end{equation}
We can first approximate the SOP using Proposition~\ref{s_sop}. Thus, $\mathcal{P}8$ can be rewritten as 
 \begin{equation}
  \begin{aligned}
  \label{p2}
\mathcal{P}9:~&\min_{\bold{\Theta}}~  P(\bold{\Theta})=\Phi\left(\frac{R-\overline{C}_{B}(\bold{\Theta})+\overline{C}_{E}(\bold{\Theta})}{\sqrt{V_{S}(\bold{\Theta})}}\right),
\\
&
s.t.~\mathcal{C}_3.
 \end{aligned}
 \end{equation}
 The challenge arises from the fact that $\delta$, $\omega$, and $\overline{\omega}$ are functions of $\theta_l$ and the non-convexity of the unimodular constraint of the phase shift. To overcome these issues, we use the gradient descent method. Specifically, in each iteration, the update of $\theta_l$ is obtained by searching in the negative gradient direction, until the value of the objective function converges to a stationary point. Next, we compute the partial derivatives with respect to $\theta_{l}$, $l=1,2,...,L$, where we use the notation $(\cdot)^{(l)}=\frac{\partial (\cdot)}{\partial \theta_{l}}$ to represent the partial derivatives. By the chain rule, the partial derivative of $P(\bold{\Theta})$ with respect to $\theta_{l}$ is given by
\normalsize
\begin{equation}
\label{de1}
\begin{aligned}
P^{(l)}(\bold{\Theta})&=\frac{\exp\left(-\frac{ T^2(\bold{\Theta})}{2} \right)T^{(l)}(\bold{\Theta})}{\sqrt{2\pi}},
\end{aligned}
\end{equation}
where 
\begin{equation}
\begin{aligned}
T(\bold{\Theta})&=\frac{R-\overline{C}_{S}}{\sqrt{V_{S}(\bold{\Theta})}},~~\overline{C}_{S}=\overline{C}_{B}(\bold{\Theta})-\overline{C}_{E}(\bold{\Theta}),
\\
\\
 T^{(l)}(\bold{\Theta})&=\frac{ -\overline{C}_{S}^{(l)}(\bold{\Theta}) V_{S}(\bold{\Theta})-\frac{1}{2}(R- \overline{C}_S(\bold{\Theta}))V^{(l)}_{S}(\bold{\Theta})}{V_{S}(\bold{\Theta})^{\frac{3}{2}}},
\\
\overline{C}^{(l)}_k(\bold{\Theta})&=\overline{\omega}_k\Tr[(\frac{1}{\delta_k} \bold{I}_{L}+\overline{\omega}_k\bold{T}_{S,k}^{\frac{1}{2}}\bold{\Theta}\bold{R}_{S}\bold{\Theta}^{H}\bold{T}^{\frac{1}{2}}_{S,k})^{-1}\bold{F}_{k,l}].
\\
\end{aligned}
\end{equation}
$\bold{F}_{k,l}$ is defined as
\begin{equation}
\bold{F}_{k,l}=\bold{T}_{S,k}^{\frac{1}{2}}(\bold{G}_{l}\otimes\bold{R}_{S})\bold{T}_{S,k}^{\frac{1}{2}},\quad l=1,2,...,L,
\end{equation}
where
\begin{equation}
\left[\bold{G}_{l}\right]_{p,q}=\left\{
\begin{aligned}
& \jmath e^{\jmath (\theta_{l}-\theta_{q})} ,&p = l, \\
& -\jmath e^{\jmath (\theta_{p}-\theta_{l})} ,&q = l, \\
&0,  &otherwise.
\end{aligned}
\right.
\end{equation}
In fact, $\bold{F}_{k,l}=(\bold{T}_{S,k}^{\frac{1}{2}}\bold{\Theta}\bold{R}_{S}\bold{\Theta}^{H}\bold{T}_{S,k}^{\frac{1}{2}})^{(l)}$ and the term $V'_{l}(\bold{\Theta})$ can be given by 
\begin{equation}
\begin{aligned}
&V_{S}(\bold{\Theta})^{(l)}=W_B^{(l)}+W_E^{(l)}-2W_{B,E}^{(l)},
\\
&W_{k}^{(l)}=
\frac{\nu_{S,k}\nu_{T,k}^{(l)}+\nu_{S,k}^{(l)}\nu_{T,k}}{\Delta_{S,k}}+ \frac{\nu_{R,k}\Gamma_{k}^{(l)}+\nu_{R,k}^{(l)}\Gamma_k}{\Delta_k},
\\
&
W_{B,E}=\frac{(\nu_{S,B,E}'\nu_{T,B,E}+\nu_{T,B,E}'\nu_{S,B,E} )}{\Delta_{S,B,E}},
\end{aligned}
\end{equation}
where
\begin{equation}
\label{gamma_de}
\begin{aligned}
\Gamma_{k}^{(l)}&=\frac{M}{L\delta^2_k}[\frac{2\nu_{T,I,k}\nu_{T,I,k}^{(l)}\nu_{S,k}}{\Delta_{S,k}} 
+\frac{\nu_{T,I,k}^2(\nu_{S,k}^{(l)}+\nu_{S,k}^{2}\nu_{T,k}^{(l)})}{\Delta_{S,k}^2} 
\\
&+ 2\omega_k \overline{\omega}_{k}^{(l)}\nu_{T,k}+\omega_k^2\nu_{T,k}^{(l)} ]-
\\
&
\frac{2M\delta_k^{(l)}}{L\delta^3_k}(\frac{\nu_{S,k}\nu_{T,I,k}^2}{\Delta_{S,k}} 
+ \omega_k^2\nu_{T,k}).
\end{aligned}
\end{equation}
The derivatives of $\nu$ in~(\ref{gamma_de}) can be given by 
\begin{equation}
\nonumber
\begin{aligned}
&\nu_{R,k}^{(l)}
=\frac{ -2M\eta_{R,l,k}(\delta_k \omega_{k}^{(l)}\overline{\omega}_k + \delta_k \omega_k\overline{\omega}_{k}^{(l)} -\omega_k\overline{\omega}_k\delta_{k}^{(l)} )}{L\delta_l^2},
\\
&\nu_{T,k}^{(l)}=-2\omega_k^{(l)}\eta_{T,k},
\\
&\nu_{S,k}^{(l)}
=-2\overline{\omega}^{(l)}_k \eta_{S,k} -2\overline{\omega}_k\eta_{S,k}(\bold{F}_{k,l})
\\
&
+\frac{2\delta_k^{(l)} \eta_{S,I,k}}{\delta^2_k}+2\nu_{S,k}(\bold{F}_{k,l}),
\\
&\nu_{S,B,E}^{(l)} =-\eta_{S,B,B,E}(\bold{F}_{B,l})\overline{\omega}_{B}  + \eta_{S,I,B,B,E} \frac{\delta_{B}^{(l)}}{\delta_{B}^2} 
\\
&
-\eta_{S,B,B,E}\overline{\omega}_{B}^{(l)}-\eta_{S,B,E,E}(\bold{F}_{E,l})\overline{\omega}_{E}  
\\
&
+ \eta_{S,I,B,E,E} \frac{\delta_{E}^{(l)}}{\delta_{E}^2} 
-\eta_{S,B,E,E}\overline{\omega}_{E}^{(l)}+\nu_{S,B,E}(\bold{F}_{BE,l}),
\\
&\nu_{T,B,E}^{(l)}=-\eta_{T,B,B,E} \overline{\omega}_{B}^{(l)}-\eta_{T,B,E,E} \overline{\omega}_{E}^{(l)},
\\
&\nu_{S,k}(\bold{F}_{k,l})=\frac{1}{M}\Tr\bold{F}_{k,l}\BG_{S,k}\FS_k \BG_{S,k},
\\
&
\eta_{S,k}(\bold{F}_{k,l})=\frac{1}{M}\Tr\BG_{S,k}\bold{F}_{k,l}\BG_{S,k}\FS_k \BG_{S,k}\FS_k 
\\
&
\eta_{S,B,B,E}(\bold{F}_{B,l})=\frac{1}{M}\Tr\BG_{S,B}\bold{F}_{B,l}\BG_{S,B}\FS_{B} \BG_{S,E}\FS_E,
\\
&\eta_{S,B,E,E}(\bold{F}_{E,l})=\frac{1}{M}\Tr\BG_{S,E}\bold{F}_{E,l}\BG_{S,E}\FS_{E}\BG_{S,B}\FS_B 
\\
&\bold{F}_{BE,l}=\LR_{S,B}(\BG_l^{H}\otimes \RT_{S,B}\RT_{S,E}\BG_l)\LR_{S,E}
.
\end{aligned}
\end{equation}
$\bold{p}_{k,l}=(\delta^{(l)}_k,\omega^{(l)}_k,\overline{\omega}^{(l)}_k)$ can be computed by~\cite[Lemma 1]{zhang2022outage},
\begin{equation}
\bold{p}_{k,l}
=\bold{A}^{-1}\bold{q}_{k,l},
\end{equation}
where $\bold{A}_k$ and $\bold{q}_{k,l}$ are defined as
\begin{equation}
\label{A_def}
\begin{aligned}
\bold{A}_k&=\begin{bmatrix}
z\nu_{R,I,k}
& \frac{M\overline{\omega}_k\nu_{R,k}  }{L} 
& \frac{M\omega_k\nu_{R,k} }{L}
\\
-\frac{\nu_{S,I,k}}{{\delta^2_k}}
&
1
&      \nu_{S,k}
\\
0 & \nu_{T,k} & 1
\\
\end{bmatrix},
\\
\bold{q}_{k,l}
&=
\begin{bmatrix}
0 & \frac{1}{M}\Tr\bold{F}_{k,l}\BG_{S,k}-2\overline{\omega}_k\nu_{S,k}(\bold{F}_{k,l}) & 0
\end{bmatrix}^{T}.
\end{aligned}
\end{equation}
The procedure of the gradient descent method is similar to the phase shifts optimization in Algorithm~\ref{AO_trans}, which is omitted here.

\section{Simulation}
\label{sec_simu}
\subsection{Simulation Settings}
Consider a uniform linear array of antennas and reflecting elements at the BS and the IRS. The correlation matrices are generated according to the model for conventional linear antenna arrays~\cite{yong2005three},
\begin{equation}
\begin{aligned}
&[\BC(d_r,\eta,\delta,N)]_{m,n}
\\
&
=\int_{-180}^{180}\frac{1}{\sqrt{2\pi\delta^2}}e^{\jmath \frac{2\pi}{\lambda} d_r(m-n)\sin(\frac{\pi\phi}{180})-\frac{(\phi-\eta)^2}{2\delta^2} }\mathrm{d}\phi,
\end{aligned}
\end{equation}
where $m$ and $n$ denote the indices of antennas and $d_r$ represents the relative antenna spacing (in wavelengths). $\eta$ and $\delta^2$ represent the mean angle and the mean-square angle spreads, which are measured by degree. $N$ is the dimension of the matrix.

The path loss of the BS-IRS link and the IRS-user (Bob or Eve) link are given by
\begin{equation}
\beta_{BS-IRS}=\frac{C_1}{d_{BS-IRS}^{\alpha_1}},~\beta_{IRS-user}=\frac{C_2}{d_{IRS-user}^{\alpha_2}},
\end{equation}
respectively. $C_i,i=1,2$ represents the reference path loss at $1$ meter and $\alpha_i,i=1,2$ denotes the path loss exponents of links. $d_{BS-IRS}$ and $d_{IRS-user}$ represent the distances. The parameters are set as $\alpha_1=2.2$, $\alpha_2=3.67$, $C_1=10^{-23.05}$, $C_2=10^{-25.95}$~\cite{kammoun2020asymptotic}, and $d_{BS-IRS}=20$ m. The noise power $\sigma^2_B=\sigma_E^2=-94$ dBm.

For the LBI case, we consider a full rank BS-IRS LoS channel matrix given by~\cite{kammoun2020asymptotic}
\begin{equation}
\begin{aligned}
&[\BH_{0}(\lambda, d_{BS},d_{IRS})]_{m,n}
\\
&
=\exp( \jmath \frac{2\pi}{\lambda}d_{BS}\sin(\theta_{LoS_{1}}(n) )\theta_{LoS_{1}}(n) 
\\
&
(n-1)d_{IRS}\sin(\theta_{LoS_2}(m))\sin(\phi_{LoS_2}(m))
  ),
  \end{aligned}
\end{equation}
where $\sin(\phi_{LoS_1}(n)$ and $\sin(\phi_{LoS_2}(m)$ are sequences of uniform angles in $[0,\pi)$ and $[0,2\pi)$, respectively and $d_{BS}=d_{IRS}=\lambda$.

\subsection{Performance Evaluation}
In Fig.~\ref{fig_out_NA} and Fig.~\ref{fig_out}, the SOP of wiretap systems and AN-aided system are given. The parameters are set as $\FR_{B}=\BC(1,0,5)$, $\FR_E=\BC(1,60,5,3)$, $\FR_{S}=\FR_{E}=\BC(1,5,5,8)$, $\FT_{B}=\BC(1,5,5,8)$, $\FT=\bold{I}_6$, $d_{IRS-B}=30$ m, and $d_{IRS-E}=40$ m. Here we aim to validate the accuracy of the analytical expressions in Section~\ref{sec_esr} and~\ref{sec_sop} and ignore the design of $\BP_W$ and $\BP_V$. Without loss of generality, we set $\BP_W=0.9\bold{I}_6,~\BP_V=0.1\bold{I}_6$ for the AN-aided system and $\BP_W=\bold{I}_6$ for the wiretap system. The number of Monte-Carlo realizations is $10^7$. Fig.~\ref{fig_out_NA} and Fig.~\ref{fig_out} demonstrate the accuracy of the approximations in Propositions~\ref{s_sop} and~\ref{an_sop}.
\begin{figure}[t!]
\centering\includegraphics[width=0.4\textwidth]{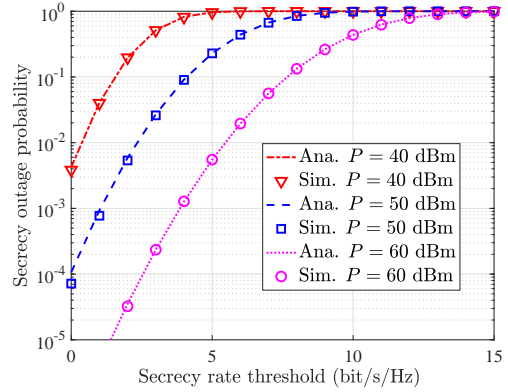}
\caption{SOP without AN.}
\label{fig_out_NA}
\end{figure}
\begin{figure}[t!]
\centering\includegraphics[width=0.4\textwidth]{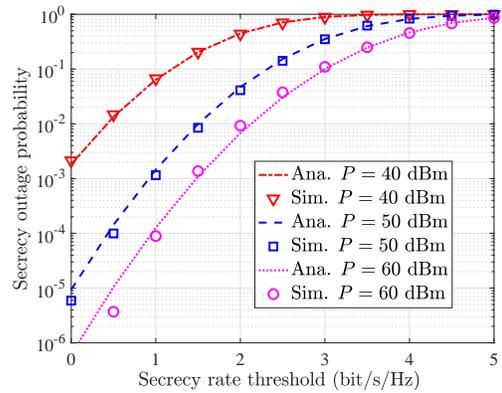}
\caption{SOP with AN.}
\label{fig_out}
\end{figure}

\begin{figure}[t!]
\centering\includegraphics[width=0.4\textwidth]{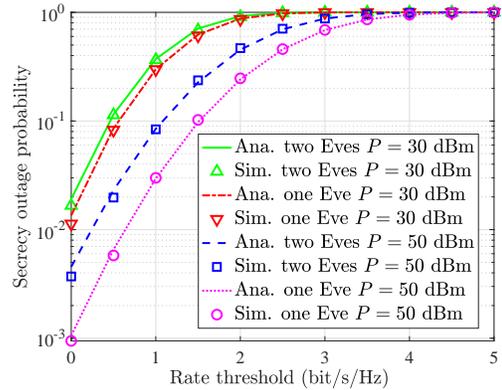}
\caption{SOP with two Eves and with AN.}
\label{mul_AN}
\end{figure}
\begin{figure}[t!]
\centering\includegraphics[width=0.4\textwidth]{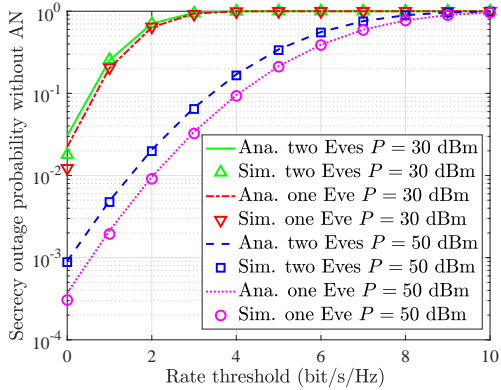}
\caption{SOP with two Eves but without AN.}
\label{mul_NA}
\end{figure}
The SOP with AN and without AN when there are two Eves are given in Fig.~\ref{mul_AN} and Fig.~\ref{mul_NA}, respectively. Here, the transmit covariance matrices are set as $\BP_{W}=0.9P\bold{I}_{M}$ and $\BP_V=0.1P\bold{I}_{M}$, where $P$ ($30$ dBm and $50$ dBm) denotes the maximum transmit power. The distances are set as $d_{IRS-B}=30$ m, $d_{IRS-E_1}=35$ m, and $d_{IRS-E_2}=35$ m. The results show that the proposed approximations are accurate for the case with two Eves. It can also be observed that more Eves result in worse SOP performance. Here, we only present the results for the double-scattering case as the LBI case has a similar phenomenon.
\subsection{Optimizations}
\begin{figure}[t!]
\centering\includegraphics[width=0.4\textwidth]{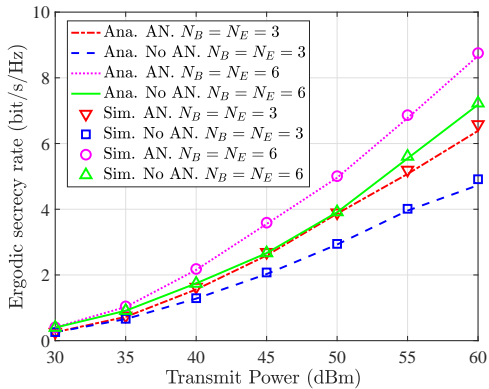}
\caption{Optimization of ESR by Algoritm~\ref{AO_trans}}
\label{esr_opt}
\end{figure}
Figure~\ref{esr_opt} depicts the ESR optimized by Algorithm~\ref{AO_trans} for the LBI case with $d_{IRS-B}=d_{IRS-E}=35$ m. It can be observed that the joint design significantly increases the ESR and the AN-aided scheme achieves better performance than that without AN. In wiretap systems, we only need to design the transmit covariance matrix and the phase shifts of the IRS while for the AN-aided system, the covariance matrix of the AN needs to be considered jointly. The AN covariance matrix provides more freedom for system design. The feasible set of the wiretap system is essentially a subset of that for the AN-aided system, so we can always find a solution for the AN-aided system that is not worse than the optimum solution of the wiretap system. Fig.~\ref{OPT_outage} illustrates the performance of the phase shifts design in~Section~\ref{opt_sop}, which indicates that the proposed scheme could decrease the SOP efficiently.

\begin{figure}[t!]
\centering\includegraphics[width=0.4\textwidth]{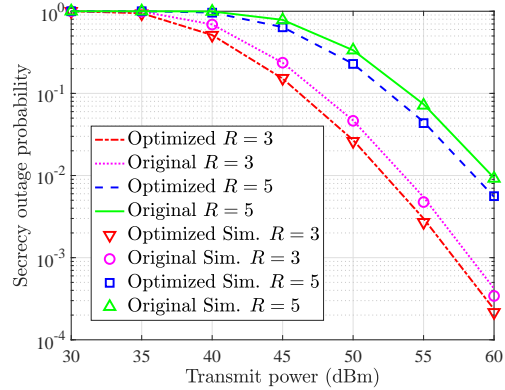}
\caption{Optimization of SOP.}
\label{OPT_outage}
\end{figure}

\section{Conclusion and Future Works}
\label{sec_con}
In this paper, by utilizing RMT, we set up new CLTs for the joint distribution of the MIs for IRS-aided MIMO secure communications, including both double-scattering and LBI channels. The Gaussianity of the joint distribution was proved by the convergence of the characteristic functions, where the closed-form expressions for the mean and covariance were also given. Based on the CLTs, the analytical expressions of the ESR and SOP for the IRS-aided MIMO communications, including both the wiretap and AN-aided systems, were derived, and the results were extended to the scenario with multiple multi-antenna Eves. Furthermore, we propose an AO algorithm to maximize the ESR of the AN-aided system by jointly optimizing the phase shifts at the IRS and the covariance matrices of the signal and the AN. Finally, a gradient algorithm was proposed to minimize the SOP for the double-scattering case. Numerical results validated the accuracy of the analytical expressions and the efficiency of the proposed optimization algorithms.

This work determined the fundamental limits of IRS-aided MIMO secure communications and provided a platform for joint optimization of the transmitter and the phase shifts. The methodology adopted in this paper is also applicable to IRS-aided systems with the direct link, which will be investigated in the future. The ESR optimization over double-scattering channel with the direct link will also be considered.

\appendices

\section{The proof of Theorem~\ref{CLT_MI}}
\label{dua_clt_proof}
The proof is motivated by the approach in~\cite{zhang2022asymptotic}, which utilizes the Gaussian tools---the integration by parts formula and Nash-Poincar{\'e} inequality~\cite{hachem2008new}. Specifically, we will investigate the characteristic function of $\bold{c}=\left(\underline{C}_{1},\underline{C}_{2},...,\underline{C}_{K}\right)^{T}$, which is given by
\begin{equation}
\label{cha_func}
\Psi(\bold{u},\bold{z})=\E e^{\jmath \bold{u}^{T}{\bold{c}} }=\E \Phi,
\end{equation}
where $\bold{u}=\left(u_{1},u_{2},...,u_{K} \right)^{T}$ and $\Phi=e^{\jmath \bold{u}^{T}{\bold{c}}}$. To show the asymptotic joint Gaussianity, we need to show that the characteristic function in~(\ref{cha_func}) converges to the characteristic function of Gaussian distribution, i.e.,
\begin{equation}
\label{cha_con}
\Psi(\bold{u},\bold{z})\xrightarrow{N\rightarrow \infty} e^{-\frac{\bold{u}^{T}\bold{M}\bold{u}}{2} },
\end{equation}
where $\BM$ represents the asymptotic covariance matrix of $\bold{c}$. Due to the difficulty in handling the logarithm of a determinant in $C_k$, we resort to handle its derivative with respect to $z_k$, i.e., the trace of the resolvent $\Tr\BQ_k=\Tr(\BH_{k}\BH^{H}_{k}+z_{k}\BI_{N})^{-1}$ and investigate convergence of the derivative for $\Psi(\bold{u},\bold{z})$ with respect to $z_k$, 
\begin{equation}
\label{ph_trace}
\frac{\partial \Psi(\bold{u},\bold{z})}{\partial z_{k}}=\jmath u_k \E \underline{ \Tr\BQ_{k}} \Phi  \xrightarrow{N\rightarrow \infty} -\frac{\bold{u}^{T}\frac{\partial \bold{M}}{\partial z_{i}}\bold{u}}{2} \Psi(\bold{u},\bold{z}).
\end{equation}

\subsection{The Evaluation of $\E \underline{ \Tr\BQ_{k}} \Phi$}
According to the resolvent identity $\bold{I}_{N_k}=z_k \BQ_k +\BQ_k \BH_k\BH_k^{H}$, we can evaluate $\jmath u_k \E \underline{ \Tr\BQ_{k}\BH_{k}\BH_{k}^{H}} \Phi$ and then further prove~(\ref{ph_trace}).
We first denote the set of indices $I_k=\{l|\BX_k=\BX_l,l=1,2,...,K\}$ to represent the $\BH$s which share the same $\BX$ with $\BH_k$, and denote $\BZ_{k}=\LR_k \BX_k \MSl_k$. We write the entries of the matrix $\BQ_{k}\BH_{k}\BH_{k}^{H}\Phi$ as the following product and use the integration by parts formula,
\begin{equation}
\label{QHHP0}
\begin{aligned}
&\E [\RTk\BY^{H}\BZ^{H}_{k} \BQ_{k}]_{j,i} [\BZ_{k}]_{n,q}  [\BY\RTk]_{p,r} \Phi
\\
&
=\sum_{p}\E [\RT_k \BY^{H}]_{j,p} [\BZ^{H}_k \BQ_k]_{p,i} [\BZ_{k}]_{n,q}  [\BY\RT_k]_{p,r} \Phi
\\
&
= \E \{\frac{1}{M} [\FT_k]_{j,r} [\BZ^{H}_k \BQ_k]_{p,i}[\BZ_k]_{i,q} \Phi
\\
&
-\frac{\Tr\BZ_{k}\BZ^{H}_{k}\BQ_{k} }{M}   [\FT_k \BH^{H}_k\BQ_k]_{j,i}[\BZ_k]_{n,q}  [\BY\RT_k]_{p,r}\Phi+
\\
&
\sum_{l}\frac{\jmath u_{l}}{M}[\RTk\RTl\BH^{H}_{l}\BQ_{l}\BZ_{l}\BZ^{H}_{k} \BQ_{k}]_{j,i} [\BZ_{k}]_{n,q}  [\BY\RTk]_{p,r}  \Phi \}.
\end{aligned}
\end{equation}
From~\cite[Theorem 1]{zhang2022asymptotic}, we have $\frac{\E\Tr\BZ_{k}\BZ^{H}_{k}\BQ_{k} }{M}=\omega_k+\BO(M^{-2})$.
By adding $ \omega_k\E [\FT_k\BH^k\BQ_k]_{j,i}[\BZ_k]_{i,q}  [\BY\RT_k]_{q,j}\Phi$ at both sides of~(\ref{QHHP0}), we can solve $\E[\RTk\BY^{H}\BZ^{H}_{k} \BQ_{k}]_{j,i} [\BZ_{k}]_{n,q}  [\BY\RTk]_{p,r} \Phi$. Summing over $j$, we can obtain
\begin{equation}
\begin{aligned}
\label{QHHP1}
&\E [\BY \FT_{k}\BY^{H}\BZ^{H}_{k} \BQ_{k}]_{p,i} [\BZ_{k}]_{n,q} \Phi=\overline{\omega}_{k}\E [\BZ^{H}_{k} \BQ_{k}]_{p,i}[\BZ_{k}]_{n,q}\Phi
\\
&-\frac{1}{M} \cov( {\Tr\BZ_{k}\BZ^{H}_{k}\BQ_{k}},[\BY\RT_k \BG_{T,k} \FT_k \BH_k^{H}\BQ_k ]_{p,i}[\BZ_k]_{n,q}\Phi
\\
&+\sum_{l}\frac{\jmath u_{l}}{M}\E [\BY\BG_{T,k}\FT_{k}\RT_{l}\BH^{H}_{l}\BQ_{l}\BZ_{l}\BZ^{H}_{k} \BQ_{k}]_{q,i} [\BZ_{k}]_{n,q}  \Phi
\\
&
+\varepsilon_{p,i,n,q},
\end{aligned}
\end{equation}
where $\varepsilon_{p,i,n,q}$ can be shown to be a $\BO(\frac{1}{N})$ term by the analysis in~\cite{zhang2022asymptotic}. By the integration by parts formula, we perform the same operations over $\E [\BZ^{H}_{k} \BQ_{k}]_{p,i}[\BZ_{k}]_{n,q}\Phi$ with respect to $[\BX_k]_{m,q}^{*}$ to obtain
\begin{equation}
\begin{aligned}
\label{QZZP}
&
\E [\BZ^{H}_{k} \BQ_{k}]_{p,i}[\BZ_{k}]_{n,q} \Phi
\\
&
=\E \sum_{m}\frac{1}{L}[ \MSr_{k}\BX^{H}\LR_{k}]_{p,m} [\BQ_{k}]_{m,i}[\BZ_{k}]_{n,q} \Phi
\\
&
 =\E \{\frac{[\MSr_{k}\MSl_{k}]_{p,q}}{L}[\FR_{k}\BQ_{k}]_{n,i} \Phi
\\
&
-\frac{  \Tr\FR_{k}\BQ_{k}}{L} [\MSr_{k}\MSl_{k}\BY\RT_{k}\BH^{H}_{k}\BQ_{k}]_{p,i}[\BZ_{k}]_{n,q} \Phi+
\\
&
\sum_{l \in I_k }\frac{\jmath u_l}{L} [\BZ_{k}]_{n,q} [\MSr_{k} \MSl_{l} \BY\RT_{l} \BH^{H}_{l} \BQ_{l}\LRl \LRk \BQ_{k} ]_{p,i}  \Phi\}.
\end{aligned}
\end{equation}
By plugging~(\ref{QZZP}) into~(\ref{QHHP1}) to replace $\E [\BZ^{H}_k \BQ_k]_{p,i}[\BZ_k]_{n,q} \Phi$ and solving $\E [\BY \FT_{k}\BY^{H}\BZ^{H}_{k} \BQ_{k}]_{p,i} [\BZ_{k}]_{n,q} \Phi$, we can obtain
\begin{equation}
\begin{aligned}
\label{HHQP}
&\E [\BY \FT_{k}\BY^{H}\BZ^{H}_{k} \BQ_{k}]_{p,i} [\BZ_{k}]_{n,q} \Phi=
\\
&
\E\{ \frac{\overline{\omega}_{k}}{L} [\BF_{S,k}\FS_k]_{p,q} [\FR_k\BQ_k]_{n,i}\Phi 
+
 \sum_{l\in I_k}\frac{\jmath u_{l} {\overline{\omega}_{k}}}{L} [\BZ_{k}]_{n,q}\times
\\
&
 [\BF_{S,k}\MSr_{k} \MSl_{l} \BY\RT_{l}  \BH^{H}_{l} \BQ_{l}\LRl \LRk \BQ_{k} ]_{q,i}  \Phi +\sum_{l}\frac{\jmath u_{l}}{M}\times
\\
 & [\BF_{S,k}\BY\BG_{T,k}\FT_{k}\RT_{l}\BH^{H}_{l}\BQ_{l}\BZ_{l}\BZ^{H}_{k} \BQ_{k}]_{q,i} [\BZ_{k}]_{n,q} \Phi \}-
\\
&
\frac{\overline{\omega}_{k}}{L} \cov( {\Tr\FR_{k}\BQ_{k}},[\BF_{S,k}\FS_{k}\BY\RT_{k}\BH^{H}_{k}\BQ_{k}]_{q,i}[\BZ_{k}]_{n,q}\Phi )
\\
&
-\frac{1}{M} \cov( [\BF_{S,k}\BY\FT^{\frac{3}{2}}_{k}\BG_{T,k}\BY\BZ^{H}_{k}\BQ_{k} ]_{q,i}
[\BZ_k]_{n,q}\Phi,
\\
&
{\Tr\BZ_{k}\BZ^{H}_{k}\BQ_{k}})
+\varepsilon.
\end{aligned}
\end{equation}
Taking $q=p$ and summing over $q$, we can solve $\E [\BZ_{k}\BY \FT_{k}\BY^{H}\BZ^{H}_{k} \BQ_{k}]_{n,i} \Phi$. Then the following equation can be obtained by the resolvent identity $\bold{I}_{N_k}=z_k \BQ_k +\BQ_k \BH_k\BH_k^{H}$, 
\begin{equation}
\label{QPII}
\begin{aligned}
&\E [\BQ_{k}]_{n,i}\Phi=[\BG_{R,k}]_{n,i}\Phi+\varepsilon_{n,i}-\sum_{l\in I_k}\frac{\jmath u_{l} \overline{\omega}_{k}}{L}\times
\\
&
\E[\BG_{R,k}\BZ_{k}\BF_{S,k}\MSr_{k} \MSl_{l} \BY\RT_l \BH^{H}_{l} \BQ_{l}\LRl \LRk \BQ_{k}]_{n,i}\Phi-
\\
&
\sum_{l}\frac{\jmath u_l}{M}\E[\BG_{R,k}\BZ_{k}\BF_{S,k}\BY\BG_{T,k}\FT_{k}\RT_{l}\BH^{H}_{l}\BQ_{l}\BZ_{l}\BZ^{H}_{k} \BQ_{k}]_{n,i}  \Phi
\\
&
+\frac{{\overline{\omega}_{k}}}{L}\cov({\Tr\FR_{k}\BQ_{k}}\Phi, [\BG_{R,k} \BZ_{k}\BF_{S,k}\FS_{k}\BY\RT_{k}\BH_{k}\BQ_{k}]_{n,i})+
\\
&\frac{1}{M}  \cov( {\Tr\BZ_{k}\BZ^{H}_{k}\BQ_{k}},[\BG_{R,k}\BZ_{k}\BF_{S,k}\BY\FT_{k}\BG_{T,k}\BH_{k}^{H}\BQ_{k} ]_{n,i}\Phi).
\end{aligned}
\end{equation}
\begin{figure*}[!htbp]
\centering
\begin{lemma} 
\label{qua_eva}
Given that $\BA$,$\BB$,$\BC$ are deterministic matrices with bounded spectral norm, the following evaluations hold true
\begin{equation}
\begin{aligned}
&
\chi_{l,k}(\BA,\BB)=\frac{1}{M}\E\Tr\BZ^{H}_{l}\BQ_{l}\BZ_{l}\BZ^{H}_{k} \BQ_{k}\bold{A}\bold{X}_{k}\bold{B},
~~\Gamma_{l,k}(\BA,\BB,\BC)=\frac{1}{M}\E\Tr\BA\BX_{k}\BB\BY\BC\BH^{H}_{l}\BQ_l \BZ_{l}\BZ^{H}_{k}\BQ_k,
\\
&
\Upsilon_{l,k}(\BA,\BB,\BC)=\frac{1}{M}\E\Tr\BA\BX_{k}\BB\BY\BC\BH^{H}_{l}\BQ_l\LR_l  \LR_k\BQ_k,
~~
\zeta_{l,k}(\BA,\BB)=\frac{1}{M}\E\Tr\BA\BX_{k}\BB\BZ^{H}_{l}\BQ_{l}\LR_l \LR_k \BQ_k,
\\
&\kappa_{k}(\BA,\BB,\BC)=\frac{1}{M}\E\Tr\BA\BX_k\BB\BY\BC\BH^{H}_k\BQ_k,
\end{aligned}
\end{equation}
\begin{equation}
\kappa_{k}(\BA,\BB,\BC)=\frac{\Tr\BF_{S,k}\MSr_k\BB}{M}\frac{\Tr\BG_{T,k}\BC\RT_k}{M}\frac{\Tr \LR_k\BA\BG_{R,k} }{L}+\BO(\frac{1}{M^2}) ,
\end{equation}
\begin{equation}
\begin{aligned}
\chi_{l,k}(\BA,\BB)\!=\!\begin{cases}
&
\frac{\delta_{l} \Tr\BA\LR_{k}\BG_{R,k}}{L}\frac{\Tr\BF_{S,k}\MSr_{k} \BB \FS_{l}\BF_{S,l}}{M\Delta_{S,l,k}\Delta_{l,k}}(1-\frac{M\overline{\omega}_{l}\overline{\omega}_{k}\nu_{R,l,k}\nu_{S,l,k}}{L\delta_{l}\delta_{k}})
+\frac{M\nu_{S,I,l,k}\Tr\BG_{R,l}\BA\LR_{l}\BG_{R,k}\LR_{k}\LR_{l}}{L^2\delta_{l}\delta_{k}\Delta_{S,l,k}\Delta_{l,k}}\frac{\Tr\BF_{S,l}\MSr_{l}\BB}{M}
\\
&
-\frac{ \Tr\BA\LR_{k}\BG_{R,k}}{L}\frac{\Tr\BF_{S,k}\FSr_{k} \BB \MSr_{l}\BF_{S,l}}{M}\frac{M \overline{\omega}_{k}\nu_{R,l,k}\nu_{S,I,l,k} }{L\delta_{l}\delta_{k}\Delta_{S,l,k}\Delta_{l,k}}
+\BO(\frac{1}{M^2})~~l \in I_{k},
\\
&\frac{\delta_{l}\Tr\LR_{k}\BA\BG_{R,k}}{L\Delta_{S,l,k}}\frac{\Tr\FS_l \BF_{S,l}\MSr_k\BB\BF_{S,k}}{M}+\BO(\frac{1}{M^2}),~~l \notin I_{k}.
\end{cases}
\end{aligned}
\end{equation}
\begin{equation}
\label{ZEBAM}
\begin{aligned}
\zeta_{l,k}(\BA,\BB)=
&\frac{\Tr\MSr_{l}\BB\BF_{S,l}}{ M}\frac{\Tr\BA\LR_l \BG_{R,l}\LR_l \LR_k \BG_{R,k} }{L\Delta_{l,k}}
+ \frac{\delta_{l}\nu_{S,I,l,k}\nu_{R,l,k}\nu_{T,l,k}\Tr\BA\LR_{k}\BG_{R,k}}{L\delta_l \delta_k \Delta_{S,l,k}\Delta_{l,k}}\frac{\Tr\BF_{S,k}\MSr_{k}\BB\FS_{l} \BF_{S,l}}{ M} , l \in I_k,
\\
\Gamma_{l,k}(\BA,\BB,\BC)&=\frac{\Tr\BG_{T,l}\RT_l \BC}{M}\chi_{l,k}(\BA,\BB)-\frac{\Tr\BG_{R,k}\LR_k \BA}{L}\frac{\Tr\MSr_k\BB \BF_{S,k}}{M}\frac{\Tr \BC\BG_{T,l}\RT_l \FT_k \BG_{T,k}}{M}\chi(\LR_k,\MSl_k)+\BO(\frac{1}{M^2}),
\\
\Upsilon_{l,k}(\BA,\BB,\BC)&=\frac{\Tr\BC\RT_{l}\BG_{T,l}}{M}\zeta_{l,k}(\BA,\BB)-\frac{\Tr\LR_{k}\BA\BG_{R,k}}{L}\frac{\Tr\MSr_k\BB \BF_{S,k}}{M}\frac{\Tr \BC\BG_{T,l}\RT_l \FT_k \BG_{T,k}}{M}\zeta(\LR_k,\MSl_k)+\BO(\frac{1}{M^2}).
\end{aligned}
\end{equation}
\end{lemma}
\vspace{-0.5cm}

\hrulefill
\end{figure*}
We introduce some important quantities and give their evaluations in Lemma~\ref{qua_eva}, whose proof can be obtained by the approach in~\cite[Appendix F]{zhang2022asymptotic},  and is omitted here due to space limitation. If we take the trace operation on both sides of~(\ref{QPII}), the RHS can be represented by $\E\Tr\FR_k \BQ_k\Phi$, $\E\Tr\BZ_k\BZ_k^H\BQ_k\Phi$, and the quantities defined in Lemma~\ref{qua_eva}. By multiplying $\MSr_k\MSl_l$ with~(\ref{HHQP}) to replace the second term in~(\ref{QZZP}), $\E\underline{\Tr\BZ_k\BZ^{H}_k\BQ_k}\Phi$ can be further written as~(\ref{QZZP_}) in next page.

By far, $\E \Tr\BQ_k\Phi$ has been represented as a linear combination of $\kappa$, $\Upsilon$, $\Gamma$, and $\E \underline{\Tr\FR_k\BQ_k}\Phi$.
Now, we only need to evaluate $\E\underline{\Tr\FR_{k}\BQ_k}\Phi$. By multiplying $\FR_k$ on both sides of~(\ref{QPII}), taking the trace operation, and replacing $\E\underline{\Tr\BZ_k\BZ^{H}_k \BQ_k} \Phi$ by $\E\underline{\Tr\FR_k \BQ_k} \Phi$, $\Upsilon$, $\Gamma$, we can solve $\E\underline{\Tr\FR_k \BQ_k} \Phi$, which is given in~(\ref{TRQP}) in next page. According to~(\ref{QPII}), $\E \Tr\BQ_k\Phi$ can be evaluated as~(\ref{TRQPHI}) in next page, where $K_{I,k}$ can be evaluated as
\begin{equation}
\begin{aligned}
K_{I,k}&=\frac{M{\overline{\omega}_k}}{L}\kappa(\LR_k\BG_{R,k},\FS_k^{\frac{3}{2}}\BF_{S,k},\RT_k)
\\
&
+ \frac{M \nu_{S,I,k} }{L\delta_k^2\Delta_{S,k}}\kappa(\LR\BG_{R,k},\FS_k^{\frac{1}{2}}\BF_{S,k},\FT^{\frac{3}{2}}_k\BG_{T,k})
\\
&
=\frac{(1-\Delta_{k}) \nu_{R,I,k}}{\nu_{R,k}} +\BO(\frac{1}{M^2}).
\end{aligned}
\end{equation}

Therefore, the problem resorts to the evaluation of $W_l$ defined in~(\ref{TRQPHI}), which can be divided into two cases:
\subsubsection{$l\in I_k$}
In this case, $W_l$ can be further computed by
\begin{equation}
\begin{aligned}
W_{l}=  V_{l,1,1}+V_{l,2,2}+\frac{K_{I,k}}{\Delta_k }( X_{l,1,1}+X_{l,2,2})
\\
+\frac{V_{l,1}+V_{l,2}+V_{l,3}+V_{l,4}-V_{l,1,1}-V_{l,2,2}}{\Delta_{k}}.
\end{aligned}
\end{equation}
The evaluations of $V_{i,j,k}$ can be done by Lemma~\ref{qua_eva}, which are given in equations (\ref{exp_v1})-(\ref{exp_v4}) at the top of the next page. Here $\eta_{S,k,k,l,a+0.5,0.5}=\frac{1}{M}\Tr\FS_k^{a}\BG_{S,k}^2\MSr_{k}\MSl_{l}\BG_l $. The results for $X_{i,j,k}$ can be computed similarly by replacing $\nu_{R,I,k}$ and $\eta_{R,I,k}$ by $\nu_{R,k,k,l}$ and $\eta_{R,k,k,l}$, respectively. The computation of the derivatives of $\delta$, $\omega$, $\overline{\omega}$ can be found in~\cite[Eq.(82)-(84)]{zhang2022asymptotic}.

\begin{figure*}
\begin{equation}
\label{QZZP_}
\begin{aligned}
&
\E\underline{\Tr\BZ_k \BZ^{H}_k \BQ_k} \Phi=
\frac{1}{\Delta_{S,k}}[(\frac{M \omega_k}{L\delta_k}-\frac{M\overline{\omega}_k \nu_{S,k} }{L\delta_{k}})
\E\underline{\Tr\FR_{k}\BQ_{k}}  \Phi
+\sum_{l \in I_k}\frac{\jmath u_l M}{L} \Upsilon_{l,k}(\LR_k,\BF_{S,k} \MSr_{k}\MSl_{l},\RT_{l} ) \E  \Phi
\\
&
-\sum_{l} \jmath u_l \delta_k  \Gamma(\LR_k, \MS_k \BF_{S,k}\MSr_{k} \MSl_{l},\FT_{k}\BG_{T,k}\RT_l )\E \Phi ]
+\BO(\frac{1}{M}).
\end{aligned}
\end{equation}
\begin{equation}
\label{TRQP}
\begin{aligned}
&
\E \underline{\Tr\FR_k \BQ_k }\Phi =\frac{1}{\Delta_{k}}[\frac{\kappa_k(\FR_k^{\frac{3}{2}}\BG_{R,k},\FS^{\frac{+}{2}}_k\BF_{S,k},\FT^{\frac{3}{2}}_k \BG_{T,k})}{\Delta_{S,k}}\sum_{l \in I_k}\frac{\jmath u_l M}{L} \Upsilon_{l,k}(\LR_k,\BF_{S,k} \MSr_{k}\MSl_{l},\RT_{l} ) 
\\
&
-\frac{\kappa_k(\FR_k^{\frac{3}{2}}\BG_{R,k},\FS^{\frac{+}{2}}_k\BF_{S,k},\FT^{\frac{3}{2}}_k \BG_{T,k})}{\Delta_{S,k}}\sum_{l} \jmath u_l \delta_k  \Gamma(\LR_k, \MS_k \BF_{S,k}\MSr_{k} \MSl_{l},\FT_{k}\BG_{T,k}\RT_l )
\\
&
-\sum_{l\in I_k }\frac{\jmath u_l {\overline{\omega}_k } M }{L}\Upsilon_{l,k}(\BG_{R,k}\FR_k^{\frac{3}{2}} ,\FS_{k}^{\frac{+}{2}}\BF_{S,k}\MSr_{k} \MSl_{l} ,\RT_l )
-\jmath u_l \Gamma_{l,k}(\BG_{R,k}\FR_k^{\frac{3}{2}},\FS_k^{\frac{+}{2}}\BF_{S,k},\BG_{T,k}\FT_{k}\RT_{l})]\E \Phi+\BO(\frac{1}{M})
\\
&=\sum_{l}\frac{\jmath u_l}{\Delta_{k}} [X_{l,1}\mathbbm{1}_{\{l\in I_k \}}+X_{l,2}+X_{l,3}+X_{l,4}\mathbbm{1}_{\{l\in I_k \}}]+\BO(M^{-1})=\sum_{l}\frac{\jmath u_l}{\Delta_{k}} X_{l}+\BO(\frac{1}{M}),
\end{aligned}
\end{equation}

\begin{equation}
\label{TRQPHI}
\begin{aligned}
&
\E \underline{\Tr \BQ_k }\Phi = K_{I,k} \E \underline{\Tr\FR_k \BQ_k }\Phi +   [\frac{\kappa_k(\FR_k^{\frac{3}{2}}\BG_{R,k},\FS^{\frac{+}{2}}_k\BF_{S,k},\FT^{\frac{3}{2}}_k \BG_{T,k})}{\Delta_{S,k}}\sum_{l \in I_k}\frac{\jmath u_l M}{L} \Upsilon_{l,k}(\LR_k,\BF_{S,k} \MSr_{k}\MSl_{l},\RT_{l} ) 
\\
&
-\frac{\kappa_k(\FR_k^{\frac{3}{2}}\BG_{R,k},\FS^{\frac{+}{2}}_k\BF_{S,k},\FT^{\frac{3}{2}}_k \BG_{T,k})}{\Delta_{S,k}}\sum_{l} \jmath u_l \delta_k  \Gamma(\LR_k, \MSl_k \BF_{S,k}\MSr_{k} \MSl_{l},\FT_{k}\BG_{T,k}\RT_l )
\\
&
-\sum_{l\in I_k }\frac{\jmath u_l {\overline{\omega}_k } M }{L}\Upsilon_{l,k}(\BG_{R,k}\FR_k^{\frac{3}{2}} ,\FS_{k}^{\frac{+}{2}}\BF_{S,k}\MSr_{k} \MSl_{l} ,\RT_l )
-\jmath u_l \Gamma_{l,k}(\BG_{R,k}\FR_k^{\frac{3}{2}},\FS_k^{\frac{+}{2}}\BF_{S,k},\BG_{T,k}\FT_{k}\RT_{l})]\E \Phi+\BO(\frac{1}{M})
\\
&=\sum_{l} K_{I,k} \E \underline{\Tr\FR_k \BQ_k }\Phi + \jmath u_l [V_{l,1}\mathbbm{1}_{\{l\in I_k \}}+V_{l,2}+V_{l,3}+V_{l,4}\mathbbm{1}_{\{l\in I_k \}}]+\BO(M^{-1})
:=\sum_{l}\jmath u_l  W_{l} +\BO(\frac{1}{M}).
\end{aligned}
\end{equation}
\hrulefill
\end{figure*}

\begin{figure*}[t!]
\begin{equation}
\label{exp_v1}
\begin{aligned}
&V_{l,1}=-\frac{M\overline{\omega}_{k}\overline{\omega}_{l}\eta_{R,I,k,k,l}\nu_{S,k,l} }{L\delta_{k}\delta_{l}\Delta_{k,l}}  
+  \frac{M\overline{\omega}_{k}\nu_{R,k,l}\nu_{R,I,k}\eta_{S,k,k,l,1.5,0.5}\nu_{T,k,l}}{L\delta_{l}^2\delta_{k}^3\Delta_{S,k,l}\Delta_{k,l}}
+
\frac{M\nu_{R,k,l}\nu_{R,I,k}\nu_{S,k,l}\overline{\omega}_{k}\overline{\omega}_{l}}{L\delta_{k}^2\delta_{l}\Delta_{k,l}}
\\
&
-\frac{M\nu_{R,k,l}\nu_{R,I,k}\eta_{S,k,k,l,1,1}\overline{\omega}_{k}\overline{\omega}_{l}}{L\delta_{k}^3\delta_{l}\Delta_{k,l}}
+\BO(\frac{1}{M^2})
=V_{l,1,1}+V_{l,1,2}+V_{l,1,3}+V_{l,1,4}+\BO(\frac{1}{M^2}).
\end{aligned}
\end{equation}
\begin{equation}
\label{exp_v2}
\begin{aligned}
 V_{l,2}& =-\frac{\nu_{T,k,l}\eta_{S,k,k,l,1,1}\nu_{R,I,k}}{\delta_{k}^2\Delta_{S,k,l}\Delta_{k,l}}(
1-\frac{M\overline{\omega}_{1}\overline{\omega}_{2}\nu_{R,k,l}\nu_{S,l,k} }{L\delta_{l}\delta_{k}}
)
-\frac{M\nu_{T,k,l}\nu_{S,I,k,l}^2\eta_{R,I,k,k,l}}{L\delta_{k}^2\delta_{l}^2\Delta_{S,k,l}\Delta_{k,l}}
+
\frac{M \nu_{R,k,l}\nu_{S,I,k,l}\nu_{S,I,k,l}\nu_{R,I,k}\nu_{T,k,l} }{L\delta_{l}^2\delta_{k}^3\Delta_{S,k,l}\Delta_{k,l}}
\\
&
-\frac{M \nu_{R,k,l}\nu_{S,I,k,l}\eta_{S,k,k,l,0.5,0.5}\nu_{R,I,k}\nu_{T,k,l} }{L\delta_{l}^2\delta_{k}^4\Delta_{S,k,l}\Delta_{k,l}}
  +\frac{1}{\Delta_{S,k,l}\delta^2_{k} }(\frac{\nu_{S,k,l}(1-\frac{M\overline{\omega}_{k}\overline{\omega}_{l}\nu_{R,k,l}}{L\delta_{k}\delta_{l}}) }{\Delta_{k,l}}
+\frac{M\nu_{S,I,k,l}^2 \nu_{R,k,l}}{L \delta_{k}^2\delta_{l}^2\Delta_{k,l}})
\\
&
\times\nu_{R,I,k}\eta_{T,k,k,l}\nu_{S,I,k}+\BO(\frac{1}{M^2})
=V_{l,2,1}+V_{l,2,2}+V_{l,2,3}+V_{l,2,4}+V_{l,2,5}+\BO(\frac{1}{M^2}).
\end{aligned}
\end{equation}
\begin{equation}
\label{exp_v3}
\begin{aligned}
&V_{l,3}=-\frac{\overline{\omega}_{k}'\Delta_{k}\eta_{S,k,k,l,2,1}\nu_{T,l,k}}{\Delta_{S,k,l}\Delta_{k,l}}(
1-\frac{M\overline{\omega}_{k}\overline{\omega}_{l}\nu_{R,l,k}\nu_{S,l,k} }{L\delta_{l}\delta_{k}}
)
-
\frac{M \overline{\omega}_{k}'\Delta_{k}\nu_{R,k,l}\nu_{S,I,k,l}\eta_{S,k,k,l,1.5,0.5} \nu_{T,k,l} }{L\delta_{k}^2\delta_{l}^2\Delta_{S,k,l}\Delta_{k,l}}
\\
&
 +\frac{{\omega}_{k}'\Delta_{k} \nu_{T,k}\eta_{T,k,k,l}\nu_{S,k}}{\Delta_{S,k,l} }(\frac{\nu_{S,k,l}(1-\frac{M\overline{\omega}_{k}\overline{\omega}_{l}\nu_{R,k,l}}{L\delta_{k}\delta_{l}}) }{\Delta_{k,l}}
+\frac{M\nu_{S,I,k,l}^2 \nu_{R,k,l}}{L \delta_{k}^2\delta_{l}^2\Delta_{k,l}})
+\BO(\frac{1}{M^2})
=V_{l,3,1}+V_{l,3,2}+V_{l,3,3}+\BO(\frac{1}{M^2}).
\end{aligned}
\end{equation}
\begin{equation}
\label{exp_v4}
\begin{aligned}
V_{l,4}
&
=
\frac{M\nu_{R,k,l}\nu_{S,k,l}\overline{\omega}_{l}\overline{\omega}_{k}'\Delta_{k}}{L\delta_{k}\delta_{l}\Delta_{k,l}}
-\frac{M\overline{\omega}_{k}\overline{\omega}_{l}\nu_{R,k,l}\eta_{S,k,k,l,2,1}\overline{\omega}_{k}'\Delta_{k}}{L\delta_{k}\delta_{l}\Delta_{k,l}}
-\frac{M\nu_{T,k,l}\eta_{S,k,k,l,1.5,0.5}\nu_{S,I,k,l}\nu_{R,k,l}\overline{\omega}_{k}'\Delta_{k}}{L\delta_{k}^2\delta_{l}^2\Delta_{S,k,l}\Delta_{k,l}}+\BO(\frac{1}{M^2})
\\
&
=V_{l,4,1}+V_{l,4,2}+V_{l,4,3}
+\BO(\frac{1}{M^2}).
\end{aligned}
\end{equation}
\hrulefill
\end{figure*}
We can verify that
\begin{equation}
\nonumber
\begin{aligned}
&V_{l,1,1}+V_{l,2,2}+\frac{K_{I,k}}{\Delta_k }( X_{l,1,1}+X_{l,2,2})
\\
&=\frac{- \partial  \log(\Delta_{k,l}\Delta_{S,k,l})}{\partial \nu_{R,k}}\frac{\partial \nu_{R,k}}{\partial z_k}+\BO(\frac{1}{M^2}),
\end{aligned}
\end{equation}
\begin{equation}
\begin{aligned}
\nonumber
&\frac{1}{\Delta_k}(V_{l,1,2}+V_{l,4,3}+V_{l,3,2}+V_{l,2,3}+V_{l,2,4}+V_{l,1,3})
\\
&
=\frac{- \partial  \log(\Delta_{k,l}\Delta_{S,k,l})}{\partial \delta_{k}}\frac{\partial \delta_{k}}{\partial z_k} 
\\
&+ \frac{- \partial  \log(\Delta_{k,l}\Delta_{S,k,l})}{\partial \nu_{S,I,l,k}}\frac{\partial \nu_{S,I,l,k} }{\partial z_k}+\BO(\frac{1}{M^2}),
\end{aligned}
\end{equation}
\begin{equation}
\begin{aligned}
\nonumber
&\frac{1}{\Delta_k}(V_{l,3,3}+V_{l,2,3})
\\
&
=\frac{- \partial  \log(\Delta_{k,l}\Delta_{S,k,l})}{\partial \nu_{T,l,k}}\frac{\partial \nu_{T,l,k} }{\partial z_k}+\BO(\frac{1}{M^2}),
\end{aligned}
\end{equation}
\begin{equation}
\begin{aligned}
\nonumber
&\frac{1}{\Delta_k}(V_{l,1,4}+V_{l,4,2}+V_{l,2,3}+V_{l,3,1})
\\
&=\frac{- \partial  \log(\Delta_{k,l}\Delta_{S,k,l})}{\partial \nu_{S,l,k}}\frac{\partial \nu_{S,l,k} }{\partial z_k}+\BO(\frac{1}{M^2}),
\\
&\frac{1}{\Delta_k}V_{l,4,1}=\frac{\partial -\log(\Delta_{k,l}\Delta_{S,l,k})}{\partial \overline{\omega}_k}  \frac{{\partial \overline{\omega}_k}}{\partial z_k}
+\BO(\frac{1}{M^2}).
\end{aligned}
\end{equation}

Therefore, when $l \in I_{k}$ we have
\begin{equation}
\label{DD_W}
W_{l}=\frac{\partial -\log(\Delta_{k,l}\Delta_{S,l,k})}{\partial z_k}+\BO(M^{-2}).
\end{equation}
\subsubsection{$l\notin I_k$}
In this case, by similar computations, we have
\begin{equation}
\label{DSW}
W_{l}=\frac{\partial -\log(\Delta_{S,l,k})}{\partial z_k}+\BO(M^{-2}).
\end{equation}
In fact, $-\log(\Delta_{k,l}\Delta_{S,l,k})$ and $-\log(\Delta_{S,l,k})$ are the covariances the two cases, respectively.
\subsection{Convergence of the characteristic function}
Define 
$[\BM]_{l,k}=-\log(\Delta_{S,l,k})-\mathbbm{1}_{l \in I_k}\log(\Delta_{k,l})$.
Then, by~(\ref{DD_W}) and~(\ref{DSW}), we can obtain
\begin{equation}
\begin{aligned}
\frac{\partial \Psi(\bold{u},\bold{z})}{\partial z_{k}} 
&=\jmath u_k \E \Tr\BQ_k \Phi
\\
&=\sum_{l}- u_l u_k \frac{\partial M_{k,l}}{\partial z_k}  \E\Phi +\BO(\frac{1}{M}),
\end{aligned}
\end{equation}
based on which we have~(\ref{cha_con}), which concludes the proof.

\ifCLASSOPTIONcaptionsoff
  \newpage
\fi



%
\bibliographystyle{IEEEtran}
\bibliography{IEEEabrv,ref}

%





\end{document}